\documentclass[final]{l4dc2024}

\title[Parameter-Adaptive AMPC]{Parameter-Adaptive Approximate MPC: Tuning Neural-Network Controllers without Retraining}
\usepackage{times}
\usepackage{paralist}
\usepackage{mathtools}
\usepackage{tikz}

\usetikzlibrary{arrows}
\usetikzlibrary{shapes}
\usetikzlibrary{snakes}
\usetikzlibrary{calc}
\usetikzlibrary{arrows}
\usetikzlibrary{shapes}
\usetikzlibrary{snakes}
\usetikzlibrary{calc}
\usetikzlibrary{tikzmark,fit}
\usetikzlibrary{positioning}
\usetikzlibrary{spy}
\usepackage{pgfplots}
\usetikzlibrary{pgfplots.statistics, pgfplots.colorbrewer, pgfplots.groupplots}
\usepackage{pgfplotstable}
\usepackage{siunitx}
\newtheorem{myremark}{Remark}

 \coltauthor{\Name{Henrik Hose} \Email{henrik.hose@dsme.rwth-aachen.de}\\
  \Name{Alexander Gräfe} \Email{alexander.graefe@dsme.rwth-aachen.de}\\
  \Name{Sebastian Trimpe} \Email{trimpe@dsme.rwth-aachen.de}\\
  \addr Institute for Data Science in Mechanical Engineering (DSME)\\RWTH Aachen University, Germany}

\begin{document}

\newcommand{\thought}[1]{{\color[rgb]{0.2,0.39,0.66}(#1)}}
\newcommand{\todo}[1]{{\color[rgb]{1.0,0.0,0.0}(#1)}}
\newcommand{\hsh}[1]{{\color{green!50!black} Henrik: #1}}
\newcommand{\st}[1]{{\color{red!50!black} Sebastian: #1}}
\newcommand{\review}[1]{{\color{red}#1}}

\newcommand{\ulm}[1]{_{\mathrm{#1}}}
\newcommand\at[2]{\left.#1\right|_{#2}}

\newcommand{\sysparam}{{\theta_\mathrm{dyn}}}
\newcommand{\objparam}{{\theta_\mathrm{obj}}}
\newcommand{\param}{\theta}
\newcommand{\paramnom}{\param\ulm{nom}}

\newcommand{\sysstate}[1]{x(#1)}
\newcommand{\sysinput}[1]{u(#1)}

\newcommand{\sysstatempc}[2]{\sysstate{#1|#2}}
\newcommand{\sysinputmpc}[2]{\sysinput{#1|#2}}

\newcommand{\mpccontr}[2]{\pi\ulm{MPC}(#1, #2)}
\newcommand{\ampccontr}[1]{\pi\ulm{NN}(#1)}
\newcommand{\aampccontr}[2]{\pi\ulm{AMPC}(#1, #2)}
\newcommand{\errorampc}[1]{e_{\pi}(#1)}
\newcommand{\ampcgrad}[1]{\nabla_{\pi\ulm{NN}}(#1)}
\newcommand{\errorampcgrad}[1]{e_{\nabla}(#1)}

\newcommand{\parameterwithunits}[5]{[\SI{#1}{\kilogram}, \SI{#2}{\kilogram}, \SI{#3}{\newton\second\per\meter}, \SI{#4}{\newton\per\volt}, \SI{#5}{\newton\meter\second\per\radian}]^\top}
\newcommand{\defaultparamvalues}{\parameterwithunits{0.02}{0.506}{-3.96}{1.3}{0.0002}}
\newcommand{\mpitunedparamvalues}{\parameterwithunits{0.01}{0.5}{-1.42}{0.5}{0.01}}
\newcommand{\quansertunedparamvalues}{\parameterwithunits{-0.02}{0.36}{0.17}{-0.82}{0.002}}
\newcommand{\simuparamvalues}{\parameterwithunits{0.04}{1}{9}{1}{0.06}}

\newcommand{\pendang}{\alpha}
\newcommand{\pendpos}{y}

\newcommand{\stagecost}[3]{\ell_{#1}(\kappa, \sysstatempc{#2}{#3}, \sysinputmpc{#2}{#3})}
\newcommand{\finalcost}[2]{V\ulm{f}(\sysstatempc{#1}{#2}, \sysinputmpc{#1}{#2})}

\newcommand{\placestacked}[2]{\stackrel{\mathclap{\normalfont\tiny\mbox{\raggedright#1}}}{#2}}

\newtheorem{assumption}{Assumption}

\newcommand{\swname}[1]{\texttt{#1}}
\newcommand{\ie}{i\/.\/e\/.,\/~}
\newcommand{\eg}{e\/.\/g\/.,\/~}
\newcommand{\cf}{cf\/.\/~}

\newcommand{\fig}{Fig\/.\/~}
\newcommand{\defn}{Def\/.\/~}
\newcommand{\sect}{Sec\/.\/~}
\newcommand{\tabl}{Tab\/.\/~}
\newcommand{\algo}{Algorithm~}
\newcommand{\theo}{Theorem~}

\newcommand{\bnnl}{3 hidden layers}
\newcommand{\bnnn}{50 neurons}
\newcommand{\bnna}{tanh activations}

\newcommand{\capt}[1]{\mdseries{\emph{#1}}}

\newcommand{\videolink}{\url{https://youtu.be/o1RdiYUH9uY}}

\newcommand{\fakepar}[1]{\vspace{1mm}\noindent\textbf{#1.}}

\maketitle
\newcommand{\ag}[1]{{\color{orange!50!black} Alex: #1}}
\begin{abstract}%
  Model Predictive Control (MPC) is a method to control nonlinear systems with guaranteed stability and constraint satisfaction but suffers from high computation times. 
  Approximate MPC~(AMPC) with neural networks (NNs) has emerged to address this limitation, enabling deployment on resource-constrained embedded systems.
  However, when tuning AMPCs for real-world systems, large datasets need to be regenerated and the NN needs to be retrained at every tuning step.
  This work introduces a novel, parameter-adaptive AMPC architecture capable of online tuning without recomputing large datasets and retraining.
  By incorporating local sensitivities of nonlinear programs, the proposed method not only mimics optimal MPC inputs but also adjusts to known changes in physical parameters of the model using linear predictions while still guaranteeing stability.
  We showcase the effectiveness of parameter-adaptive AMPC by controlling the swing-ups of two different real cartpole systems with a severely resource-constrained microcontroller (MCU).
  We use the same NN across both system instances that have different parameters.
  This work not only represents the first experimental demonstration of AMPC for fast-moving systems on low-cost MCUs to the best of our knowledge, but also showcases generalization across system instances and variations through our parameter-adaptation method.
  Taken together, these contributions represent a marked step toward the practical application of AMPC in real-world systems.
\end{abstract}

\begin{keywords}%
  Nonlinear model predictive control, approximate MPC, sensitivity, machine learning
\end{keywords}

\section{Introduction}\label{sec:introduction}
Model predictive control (MPC) is an optimization-based control strategy for nonlinear systems with favorable properties like theoretically guaranteed stability and constraint satisfaction \citep{rawlings2017model}.
These properties come at the cost of long computation times for solving nonlinear programs, even on modern, high-performance processors.
In practice, this limits applications of~MPC, as powerful and expensive hardware is not always available.
Therefore, in recent years, approximate MPC (AMPC) with neural networks (NNs) has gained increasing interest~\citep{gonzalez2023neural}.
In AMPC, the behavior of the MPC is imitated by a NN that is often small enough to be evaluated within milliseconds on a microcontroller (MCU).
This unlocks the potential of~MPC on small, resource-constrained, embedded systems.
However, few publications have demonstrated~NN approximations of MPCs for fast-moving systems with update rates at the order of milliseconds on low-cost MCUs with real hardware experiments~(\cf\sect\ref{sec:intro-embeddedNN}).
A particular challenge when deploying AMPC in real-world applications is tuning the AMPC on the physical system to achieve satisfactory performance.
This usually involves changing parameters in the MPC (\eg system model, cost function) and retraining the NN approximation.
However, every tuning iteration potentially takes days to recompute the dataset.
This is because the dataset easily contains tens of thousands to millions of samples from the MPC where each datum corresponds to one solution of the MPC optimization.
As a consequence, the whole tuning process can take weeks.
While this might be cumbersome yet possible in a research context~\citep{nubert2020safe, carius2020mpc, abu2022deep, leonow2023embedded}, it limits practical applicability of AMPC in real products.
This is even more prohibitive when parameters vary slightly for each system instance and the AMPC needs to be retuned every time~\citep{adhau2019embedded}.

\fakepar{Contribution}
In this paper, we address the issue of cumbersome AMPC tuning by proposing a method that allows for tuning without regenerating a dataset and without retraining the NN.
We achieve this using local sensitivities of nonlinear programs with respect to its parameters.
The proposed AMPC architecture (see~\fig\ref{fig:method}) not only imitates the optimal inputs of the MPC, but also its sensitivities.
The architecture allows us to adjust the controller's output online with linear predictions.
Thus, the AMPC accounts for time-invariant, known changes in parameters of the dynamics model such as mass or friction, and also others like cost or constraint parameters of the~MPC.
In summary, we make the following contributions to the field of AMPC.
\begin{compactenum}
\item We propose an AMPC architecture which can adjust the output of the trained NN to time-invariant, known changes in dynamics and MPC design parameters using a linear predictor with approximations of the local sensitivities.
\item We derive conditions for the approximation error and maximum parameter change under which the controlled system is stable.
\item We implement the method on a small MCU (STM32G474 running at 170MHz with 96kb of SRAM) that costs only a few dollars and show that it can swing up and balance a real cartpole system.
To the best of our knowledge, this is the first real-world implementation of AMPC on a severely resource-constrained, general-purpose MCU controlling a fast and unstable nonlinear system.
\item In hardware experiments, we demonstrate that the parameter-adaptive AMPC generalizes to two cartpole system instances with different parameters without retraining the NN, whereas a naive nominal AMPC fails to generalize.\\
\end{compactenum}
A video of our experiments is available at~\videolink.

\begin{figure}
	\centering
	\begin{tikzpicture}[thick]
    \newcommand{\distsouround}{0.1cm}

    \pgfdeclarelayer{bg}    %
    \pgfsetlayers{bg,main}  %

    \begin{scope} [draw = black,
        fill = white, 
        dot/.style = {black, radius = .05}]

        \node (contr) at (0,0) [draw, fill, text width=6cm] {NN approx.~optimal inputs: $\ampccontr{x}$};
        \node [below=0.2cm of contr, draw, fill, text width=6cm] (sens) {NN approx.~sensitivities: $\ampcgrad{x}$}; 
        \node [above left=0.1cm of contr, anchor=south west, text width=8cm] (labelcontr) {\textbf{Parameter-adaptive AMPC:} $\aampccontr{x}{\param}$};

        \node [circle, right = 0.3cm of sens, draw, fill=white, scale=0.7] (senstimes) {$\times$}; 

        \node [below = 0.cm of sens] (params) {Change in parameter: $\param-\paramnom$};

        \draw[->] (sens) to (senstimes);
        \draw[->] (params) -| (senstimes);

        \coordinate (mid) at ($(contr)!0.5!(sens)$);
        
        \node [right=4cm of mid, circle,fill=white, draw, scale=0.7] (contradd) {$+$};

        \draw[->] (contr) -| (contradd);
        \draw[->] (senstimes) -| (contradd);

        \node [right=1.5cm of contradd, fill=white, align=center, minimum height = 1cm, text width=3cm, draw] (sys) {\textbf{System}\\$\sysstate{k+1}=f_\sysparam(\sysstate{k}, \sysinput{k})$};

        \draw[->] (contradd) -- node[above, pos=0.6] {$\sysinput{k}$} (sys);

        \coordinate[above left= \distsouround of labelcontr] (lefttopedge);
        \coordinate[right= \distsouround of contradd] (righttopedge);
        \coordinate[below = \distsouround of params] (bottomedge);

        \coordinate[below = 1.6cm of righttopedge] (contrlabel);

        \coordinate[right=0.3cm of sys] (a);
        \coordinate[below=0.1cm of bottomedge] (b);
        \coordinate[left=0.7cm of sens] (c);

        \draw[-] (sys) -- (a);
        \draw[-] (a) |- (b);
        \draw[-] (b) -| (c);

        \draw[->] (c) |- node [midway, left] {$\sysstate{k}$} (contr);
        \draw[->] (c) |- (sens);

        \begin{pgfonlayer}{bg} 
            \draw[draw=black,fill=lightgray!20,thick] (-3.5, 1.2) rectangle ++(8.3,-3);
        \end{pgfonlayer}

    \end{scope}
\end{tikzpicture}
	\vspace{-0.2cm}
	\caption{Parameter-adaptive AMPC. \capt{Approximate nominal MPC inputs are linearly adapted to true parameters $\param$ by approximate sensitivities. The parameters can include parameters of the dynamics model $\sysparam$ and other MPC parameters, like weights of the cost function.}}
	\label{fig:method}
\end{figure}
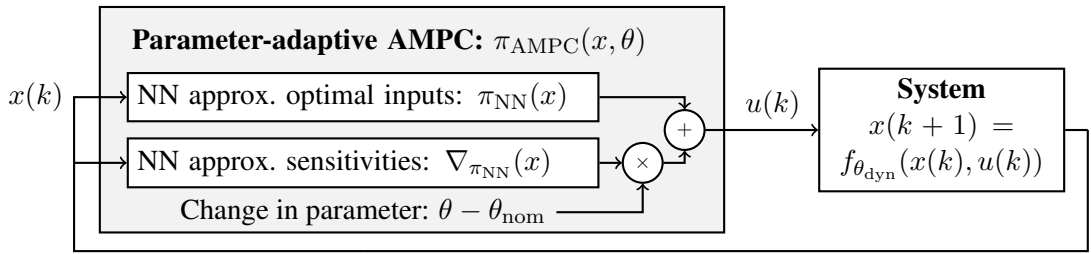

\section{Related Work}\label{sec:related-work}
This section summarizes essential related work in the context of this paper, divided into three parts.
The first part discusses important methods for AMPC and existing implementations of AMPC in embedded systems, highlighting the limited usage of small MCUs for controlling fast and nonlinear real systems.
The second part summarizes the few AMPC methods that could include dynamics model variations, highlighting the general lack of parameter-adaptive AMPC methods.
The final part reviews prior works on MPC that leverage sensitivities, but for different purposes than dynamics model parameter adaptation as herein. 

\subsection{Approximate model-predictive control}\label{sec:intro-embeddedNN}
In the past decades, various methods have been proposed to avoid online optimization in embedded~MPC.
In simple lookup methods, an offline-computed table of optimal inputs is used for online interpolation.
Even using the nearest neighbor can perform comparable to imitation learning methods~\citep{florence2022implicit} and can have theoretical guarantees for robust MPC~\citep{bayer2016tube}.
Through multi-parametric programming, explicit MPC solutions are attainable exactly for (small) linear~\citep{bemporad2002explicit, alessio2009survey} and approximately for some nonlinear systems~\citep{johansen2004approximate, bemporad2006algorithm}, for which convenient software packages for stand-alone C-code export for MCUs exist~\citep{kvasnica2015design}.
However, solving multi-parametric programs is complex for general nonlinear and even medium-sized linear systems and lookup tables grow exponentially with the number of optimization variables, making them impractical for MCUs with limited memory.

Possible alternatives that efficiently scale with system dimensions include approximations with NNs, a concept explored for decades~\citep{parisini1998nonlinear, aakesson2006neural}. 
More recently, research focused on theoretical guarantees, imitation learning methods, and validation of NN approximations, \cf a recent overview by~\citet{gonzalez2023neural}.
However, only a handful of publications have successfully implemented NN approximations of MPCs in practical hardware experiments, such as robot arm tracking~\citep{nubert2020safe}, electric motor control~\citep{abu2022deep}, or quadruped robot control~\citep{carius2020mpc}; notably, all using laptop-grade CPUs for inference, which would, in principle, also be suitable for real-time optimization.

On the other side, works focussing on the implementation of AMPC on small embedded devices validated their results on simulated plants \citep{lucia2020deep, wang2021model, chan2021deep, adhau2019embedded, karg2019learning}.
To the best of our knowledge, \citet{leonow2023embedded} and \citet{xiang2024use} present the only applications of AMPC with NNs on MCUs controlling actual physical systems, albeit with slow dynamics and a control frequency of \SI{2}{\hertz} in the first and a linear system in the latter case.
In this paper, we implement and experimentally test AMPC (evaluated in less than~\SI{2}{\milli\second}) on MCUs controlling a fast physical system (control frequency of \SI{20}{\hertz}).

\subsection{AMPC with parameter variation}
When bringing control to real physical systems, controllers usually need to be fine-tuned on the real plant to achieve optimal performance.
In classic MPC, this can be done automatically~\citep{forgione2020efficient,paulson2023tutorial} in closed-loop experiments.
For AMPC, however, such tuning involves recomputing a large dataset and retraining an NN, a limitation specifically mentioned by~\citet{adhau2019embedded}.
While AMPC can imitate a robust MPC policy \citep{nubert2020safe} that could also account for dynamics model parameter variations~\citep{kohler2020computationally}, such robustification schemes are generally involved for nonlinear systems and introduce conservatism.
Other AMPC methods that were not originally intended for tuning could be used to adapt to dynamics parameter changes, \eg warm-starting online optimization~\citep{klauvco2019machine, chen2022large} or retraining the NN in a reinforcement-learning fashion, similar to~\cite{bogdanovic2022model}.
However, implementing them on small MCUs is impractical and online optimization with NN warm-starting can be slow \citep{vaupel2020accelerating}.
We overcome this issue by using sensitivities of nonlinear programs to approximately correct the AMPC with a linear prediction.
The presented method is fast to compute on MCUs and intuitive to tune due to physical interpretation of dynamics model parameters.

\subsection{Sensivitiy-based nonlinear MPC}\label{sec:intro-sensNLP}
The method proposed in this paper uses the sensitivities of the nonlinear MPC problem to predict online how the optimal inputs change in response to parameter variations.
In this section, we briefly introduce sensitivities of nonlinear programs (NLPs) first and then summarize some related works that leverage sensitivities for fast adaptations of MPCs to changes in initial state.

A nonlinear MPC problem is a parametric NLP with parameters~$p$.
In sensitivity MPC, these parameters are typically the initial state~\citep{buskens2001sensitivity} but could also comprise other dynamics model parameters as in this paper, or even cost or constraint parameters.
A parametric NLP in standard form with cost function~$V$, constraints~$g$, and optimization variables~$u$ is
\begin{align}\label{eqn:nlp-standard}
	\begin{split}
	u^*(p) = \arg\min_{u}\; V(u,p) \quad \text{s.t.} \quad g(u, p)\leq0.
	\end{split}
\end{align}

We call the constraints for which $g(u^*, p) = 0$ active constraints.
The classic results of~\cite{fiacco1976sensitivity} allow applying the implicit function theorem to the gradient of the KKT conditions, which can be used to get first order predictions of the optional solution
\begin{equation}
	u^*(p_0 + \Delta p) = u^*(p_0) + \at{\tfrac{\partial}{\partial p}u^*(p)}{p = p_0}\Delta p + \mathcal{O}(\|\Delta p\|^2).
\end{equation}
This holds as long as the set of active constraints does not change, in which case a quadratic program would need to be solved~\citep{kadam2004sensitivity}.
The gradient~$\tfrac{\partial u^*}{\partial p}$ can be used as a linear predictor of the optimal solution to (\ref{eqn:nlp-standard}) for parameter variations.
Sensitivities are available in some nonlinear optimization software packages used for MPC,~\eg Pyomo~\citep{bynum2021pyomo}, CasADi~\citep{andersson2018sensitivity}, or acados~\citep{verschueren2022acados}.

The sensitivity with respect to the initial state has been used extensively in MPC.
Classic sensitivity MPC uses it to perform a fast online correction for a precomputed input sequence~\citep{buskens2001sensitivity,diehl2002real}; also when active constraints change~\citep{kadam2004sensitivity}.
In the advanced-step MPC~\citep{zavala2009advancedstepmpc,yang2013advancedmultistepmpc,jaschke2014fast}, a nonlinear MPC is solved online for a predicted future initial state.
Once numerical optimization concludes, the solution is adjusted to the measured initial state using sensitivities.

\cite{krishnamoorthy2021sensitivity,krishnamoorthy2023improved} uses sensitivities with respect to the initial state in AMPC for data augmentation. 
Linear (offline) predictions around sample points extend the dataset to enhance the accuracy of the NN at a marginal increase in computation times.
Similarly,~\cite{lueken2023sobolev} use the sensitivities directly in the training with a loss on the gradient to improve prediction accuracies.

While these works indicate that sensitivities can be used for fast online adjustment to changes in the initial state, none of them uses sensitivities to account for changes in dynamics and other MPC parameters as we propose herein.

\section{Parameter-adaptive AMPC}\label{sec:methods}
This section first introduces the nonlinear MPC formulation that is approximated.
We then describe the proposed parameter-adaptive AMPC and provide theoretical results on its stability.

We consider general, nonlinear, discrete time dynamical systems
\begin{equation}\label{eqn:nominal_sys}
	\sysstate{k+1} = f_{\sysparam}(\sysstate{k}, \sysinput{k}), ~~\sysstate{0} = x_0,
\end{equation}
where
$k\in \mathbb{N}$ is the discrete time,
$x\in\mathbb{R}^n$ the state,
$x_0\in\mathbb{R}^n$ the initial state,
$u\in\mathbb{R}^m$ the input,
and $f_\sysparam: \mathbb{R}^n\times\mathbb{R}^m \rightarrow \mathbb{R}^n$ the continuous dynamics function with time-invariant dynamics model parameters $\sysparam\in\Theta_\mathrm{dyn}\subseteq\mathbb{R}^{q_\text{dyn}}$,
where $\Theta_\mathrm{dyn}$ is the set of possible parameter values.
We consider the dynamics parameters to be time invariant, representing, \eg different instances of the same system class.
We consider a nominal system~(\ref{eqn:nominal_sys}) without external disturbances, as is often done in practice, \eg for approximate MPC by~\cite{carius2020mpc}.
However, the same parameter-adaption method would also extend to robust AMPC with external disturbances, \eg by~\cite{nubert2020safe}.

\subsection{Input robust model predictive control}
A standard, nonlinear MPC formulation~\citep{rawlings2017model} for controlling~(\ref{eqn:nominal_sys}) is
\begin{align}
	\begin{split}\label{eqn:implicit-mpc}
	&u_\param^* = \arg \min_{u}{\textstyle\sum_{\kappa=0}^{N}}\stagecost{\param}{\kappa}{k}\\
	\text{s.t.~~} & \sysstatempc{0}{k} = \sysstate{k}, \quad \sysstatempc{\kappa+1}{k} = f_\sysparam(\sysstatempc{\kappa}{k}, \sysinputmpc{\kappa}{k}),\\
	& \sysstatempc{\kappa}{k} \in \mathcal{X}_{\param}(\kappa), \quad \sysinputmpc{\kappa}{k} \in \mathcal{U}_{\param}(\kappa) \quad \forall\kappa=0\dots N,
	\end{split}
\end{align}
where $\ell_{\param}: \mathbb{N}\times\mathbb{R}^n\times\mathbb{R}^m \rightarrow \mathbb{R}$ is a cost function, and $\mathcal{X}_\param(\kappa)\subseteq\mathbb{R}^n$ and $\mathcal{U}_\param(\kappa)\subseteq\mathbb{R}^m$ are the state and input constraints, respectively. 
The parameter vector $\param\in\Theta\subseteq\mathbb{R}^q$ contains the parameters of the dynamical system $\sysparam$ and possibly other parameters, like weights in the objective function.
The set $\Theta$ is the set of all possible parameter values.
The first element~$u_\param^*(0|k)$ of the optimal input trajectory $u_\param^*$ is then applied to the plant.
We call this controller $\mpccontr{\sysstate{k}}{\param} := u_\param^*(0|k)$.
In practice, there are software packages to formulate and pass the optimization problem~(\ref{eqn:implicit-mpc}) to NLP solvers to compute the optimal control trajectory $u^*$.
Some software packages (\cf\sect\ref{sec:intro-sensNLP}) also provide sensitivities for the optimal controls, \ie gradients~$\at{\tfrac{\partial}{\partial \param}\mpccontr{x}{\param}}{\paramnom}$, which we shall use in the following.
We further assume that the MPC~(\ref{eqn:implicit-mpc}) stabilizes
the plant under input disturbances $d(k)\in\mathbb{R}^m$, \ie $\sysinput{k} = \mpccontr{\sysstate{k}}{\param} + d(k)$.
The input disturbances model approximation errors of the AMPC, not external disturbances or noise (note that~(\ref{eqn:nominal_sys}) is the nominal system), as similarly proposed by~\cite{hertneck2018ampc}.
The general formulation of an MPC~(\ref{eqn:implicit-mpc}) includes MPCs with or without terminal state costs and constraints.
We do not further specify the structure of the MPC, a potential robustification scheme, cost function, and constraints.
These depend on the application and their choice is orthogonal to the proposed approximation method.
For the following derivations, we require a robust MPC:
\begin{assumption}
	\label{as:robuststability}
	There exists a stability notion, a vector norm, a compatible matrix norm $||\cdot||$, and a maximum allowable input disturbance $\eta > 0$ such that it holds for all $\param\in\Theta$: if for all $k \geq 0$~$||d(k)||\leq\eta$, the controlled system is stable under this notion.
\end{assumption}
An overview of robust MPC methods that can satisfy this assumption is presented by~\cite{houska2019robust}.
In the context of AMPC,~\cite{hertneck2018ampc} and~\cite{nubert2020safe} use an input-robust MPC (\cf\cite{kohler2020computationally} for a more general version) that can, in principle, satisfy this assumption.

\subsection{Approximate model predictive control with parameter changes}
A general idea of AMPC is to approximate~$\pi\ulm{MPC}$ via supervised learning from a large dataset of samples~$(x_i, \pi\ulm{MPC}(x_i, \param))$.
As the MPC depends on~$\param$, we consider a nominal~$\paramnom\in\Theta$ for training. 
For example,~$\paramnom$ could be chosen as the center of the relevant set~$\Theta$.
The learned mapping for nominal dynamics parameters is called $\pi\ulm{NN}: \mathbb{R}^n\rightarrow\mathbb{R}^m$.
The approximation error of~$\pi\ulm{NN}$ is $\errorampc{x} := \ampccontr{x}-\mpccontr{x}{\paramnom}$.
In addition to approximating the optimal controls by the MPC, as is standard in AMPC, we propose to train an additional NN offline via supervised learning to approximate the sensitivities from samples~$(x_i, \at{\tfrac{\partial}{\partial \param}\mpccontr{x_i}{\param}}{\paramnom})$.
These sensitivities are available from some NLP solvers (\cf\sect\ref{sec:intro-sensNLP}).
We call the NN that approximates sensitivities~$\ampcgrad{x}:\mathbb{R}^{n}\rightarrow\mathbb{R}^{m\times p}$.
The approximation error of this is $\errorampcgrad{x} := \ampcgrad{x} - \at{\tfrac{\partial}{\partial \param}\mpccontr{x}{\param}}{\paramnom}$.
With the two NN approximations,~$\ampccontr{x}$ and~$\ampcgrad{x}$, we propose the parameter-adaptive AMPC (see~\fig\ref{fig:method}):
\begin{equation}
	\label{eq:aampc}
	\aampccontr{x}{\param} = \ampccontr{x} + \ampcgrad{x}(\param - \paramnom).
\end{equation}
This parameter-adaptive AMPC (\ref{eq:aampc}) uses the approximated optimal controls for the nominal parameters and linearly predicts how these controls change for the real parameters~$\param$.
The linear predictor is given by the approximate sensitivities~$\ampcgrad{x}$.

In the following, we will derive stability properties of the parameter-adaptive AMPC~(\ref{eq:aampc}).
Based on the classic result on NLP sensitivities by~\cite{fiacco1976sensitivity}, there exists a constant $L$ such that
\begin{equation}\label{eqn:predictor-lip}
	||\mpccontr{x}{\paramnom} + \at{\tfrac{\partial}{\partial \param}\mpccontr{x}{\param}}{\paramnom}(\param - \paramnom) - \mpccontr{x}{\param}|| \leq L||\param - \paramnom||,
\end{equation}
if the set of active constraints does not change \citep{zavala2009advancedstepmpc}.
We make this assumption:
\begin{assumption}\label{ass:active-set}
	The set of active constraints does not change by the linear predictor. That is, the solution to~(\ref{eqn:implicit-mpc}) has the same active constraints for~$\theta$ and~$\paramnom$.
\end{assumption}
Using~(\ref{eqn:predictor-lip}), we can extend the stability of AMPC as per Ass.~\ref{as:robuststability} to the parameter-adaptive AMPC~(\ref{eq:aampc}):
\begin{theorem}
	\label{th:stability}
	Let Assumptions \ref{as:robuststability} and \ref{ass:active-set} hold.
	If there exists an $\epsilon>0$ such that $\errorampc{x} + \epsilon < \eta$, then the dynamical system $f_\sysparam$ controlled by (\ref{eq:aampc}) is stable for all $\param\in\tilde{\Theta}$ with $\tilde{\Theta}:=\{\theta\in\Theta|(\sup_{x\in\mathbb{R}^n}||\errorampcgrad{x}|| + L)||\param - \paramnom||<\epsilon\}$.
\end{theorem}

\begin{proof}
	Because the set of active constraints does not change, we have
	\begin{align}
		||\pi&\ulm{AMPC}(x,\param) - \mpccontr{x}{\param}|| = ||\ampccontr{x} + \ampcgrad{x}(\param - \paramnom)-\mpccontr{x}{\param}||\nonumber\\
		=&||\errorampc{x} + \mpccontr{x}{\paramnom} + \big[\at{\tfrac{\partial}{\partial \param}\mpccontr{x}{\param}}{\paramnom}+\errorampcgrad{x}\big](\param - \paramnom) - \mpccontr{x}{\param}||\nonumber \\
		\leq&||\errorampc{x}|| +||\errorampcgrad{x}(\param - \paramnom)||\nonumber \\&
		+||\mpccontr{x}{\paramnom} + \at{\tfrac{\partial}{\partial \param}\mpccontr{x}{\param}}{\paramnom}(\param - \paramnom) - \mpccontr{x}{\param}||\nonumber\\
		\leq&||\errorampc{x}|| + ||\errorampcgrad{x}(\param - \paramnom) || + L||\param - \paramnom||\nonumber\\
		\leq&\eta-\epsilon + (||\errorampcgrad{x}|| + L)||\param - \paramnom||.\nonumber
	\end{align}
	Hence, using Assumption \ref{as:robuststability}, the theorem directly follows.
\end{proof}%
Theorem~\ref{th:stability} provides a sufficient condition that assures that the AMPC for the nominal parameters~$\paramnom$ generalizes locally to a stabilizing controller for~$\param$ in the set~$\tilde{\Theta}$ using the gradient of the nominal solution.
The size of the environment depends on the accuracy of the learned MPC inputs and gradients. 
In general, the higher the accuracy, the larger the environment.
For the gradient, this dependency follows directly from the formula of $\tilde{\Theta}$.
For a more accurate $\ampccontr{x}$, $\epsilon$ can be chosen higher, and thus the size of $\tilde{\Theta}$ increases.
In practice, global approximation error bounds
  could be validated, e.g., by statistical methods~(\cite{hertneck2018ampc,nubert2020safe}).
\begin{myremark}
	In case the set of active constraints changes, accurate sensitivities require solving a quadratic program \citep{kadam2004sensitivity}, or can be avoided by heuristics like soft constraints \citep{yang2013advancedmultistepmpc}.
	As we will empirically show in \sect\ref{sec:results}, neglecting active set changes can still work well in practice for small changes in the parameters.
\end{myremark}

\section{Hardware Implementation: Cartpole Pendulum}\label{sec:implementation}
This section describes a practical implementation example of parameter-adaptive AMPC with a cartpole pendulum swing-up\footnote{Code available at:~\url{https://github.com/hshose/Adaptive-AMPC-Cartpole}}.
First, we systematically test the parameter-adaptive AMPC for parameter variations in simulation.
Second, we present hardware experiments to underline our method's practical relevance.
We use two different instances of a cartpole pendulum (\fig\ref{fig:cartpole}), one produced by Quanser Inc.~\citep{apkarian2012quanser}, one self-made~\citep{mager2022scaling}.
The parameter-adaptive AMPC transfers from nominal parameters to both real system instances without retraining.
In both hardware systems, just using approximate controls~$\pi\ulm{NN}$ for nominal parameters fails.

The inverted pendulum is a classic benchmark control system~\citep{boubaker2013inverted}.
The standard cartpole model~(\cf~\fig\ref{fig:cartpole}) has the state~$x=[\pendpos,\dot{\pendpos},\pendang,\dot{\pendang}]^\top$, which contains the pendulum angle, cart positions, and corresponding derivatives.
The input~$u$ is the voltage applied to the cart's motor.
Neglecting the dynamics of the motor current\footnote{This is reasonable for small motor inductances in the hardware~\citep{apkarian2012quanser}.}, the force moving the cart is~$F=C_1\dot{\pendpos}+C_2u$, where the constants~$C_1$ and~$C_2$ capture all velocity-dependent friction as well as motor and transmission constants, respectively.
The inertia~$J = J_\text{rod}+J_{m_\text{add}}$ and the mass~$m=m_\text{rod} + m_\text{add}$ of the pendulum are given about its center of mass at a distance~$l$ from the point of rotation.
The mass of the cart~$M=M_\text{cart}+J_\text{mot}$ includes the reflected inertia of the motor and transmission.
Using Euler-Lagrange's equations, we get the equations of motion
\begin{align}\label{eqn:eom-implicit}
    \begin{split}
m l \cos(\pendang) \ddot{\pendpos} + (ml^2+J)\ddot{\pendang} - mgl\sin(\pendang) &= - C_3\dot{\pendang} 
\\
(M+m)\ddot{\pendpos} + m l \cos(\pendang)\ddot{\pendang} - ml\sin(\pendang)\dot{\pendang}^2 &= F
,
\end{split}
\end{align}
where constant~$C_3$ is the velocity-proportional friction coefficient of the pendulum rotational axis.

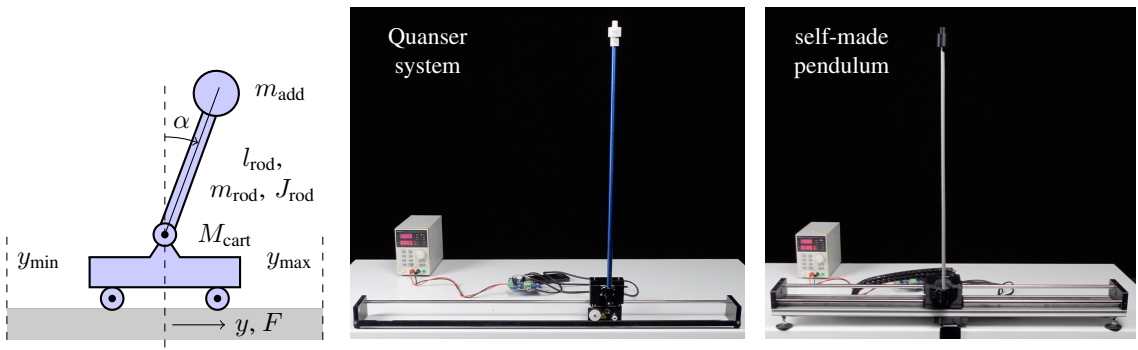
\begin{figure}[hbt]
    \center
    \begin{tikzpicture} [thick]

    \tikzstyle{every node}=[font=\small]
    
    \newcommand{\angpendfig}{20}
    
    \draw [black] (-2.1,0) -- (2.1,0);
    \fill [black!20] (-2.1,0) rectangle (2.1,-.4);
    
    \begin{scope} [draw = black,
        fill = blue!20, 
        dot/.style = {black, radius = .05}]
    
    \filldraw [rotate around = {-\angpendfig:(0,1)}] (.1,1) --
        +(0,2) arc (0:180:.09)
        coordinate [pos = .5] (T) -- (-.1,1);

    \node at (1.3, 1.8) [text width=2.5cm, align=center] {$l_\text{rod}$, \\$m_\text{rod}$, $J_\text{rod}$};
    
    \filldraw [rotate around = {-\angpendfig:(0,1)}](0,3.) circle (0.3) +(0.8,0.3) node {$m_\text{add}$};
    
    \filldraw (-.7,.15) circle (.15);
    \fill [dot] (-.7,.15) circle;
    \filldraw (.7,.15) circle (.15);
    \fill [dot] (.7,.15) circle;

    \filldraw (-1,0.7) -- coordinate [pos = .5] (F)
    (-1,.3) --
    (1,.3) -- (1,0.7) 
    coordinate (X) -- (.2,0.7)
    -- (0,1) -- (-.2,0.7) -- (-1.014,0.7);
    \node at (0.8,1.0) {$M_\text{cart}$};

    \filldraw   (0,1) circle (.15);
    \fill [dot] (0,1) circle;
    \end{scope}
    
    \begin{scope} [thin, black]
        \draw (T) -- (0,1.) coordinate (P);
        \draw [dashed] (P) + (0,-1.5) -- +(0,2);
        \draw [->] (P) + (0.1, -1.2) -- +(0.8, -1.2) node [black, right] {$\pendpos$, $F$};
        \draw [->] (P) + (0,1.3) arc (90:90-\angpendfig:1.3) node [black, midway, above] {$\pendang$};
        \draw [dashed] (-2.1,-0.4) -- (-2.1,1) node [black, near end, right] {$\pendpos_\text{min}$};
        \draw [dashed] ( 2.1,-0.4) -- ( 2.1,1) node [black, near end, left] {$\pendpos_\text{max}$};
    \end{scope}
    
    \end{tikzpicture}
    \hspace{-0.6cm}
    \begin{tikzpicture}
        \node (firstpic) at (0,0) {\includegraphics[clip, trim=0cm 0cm 0cm 0cm, height=4.4cm]{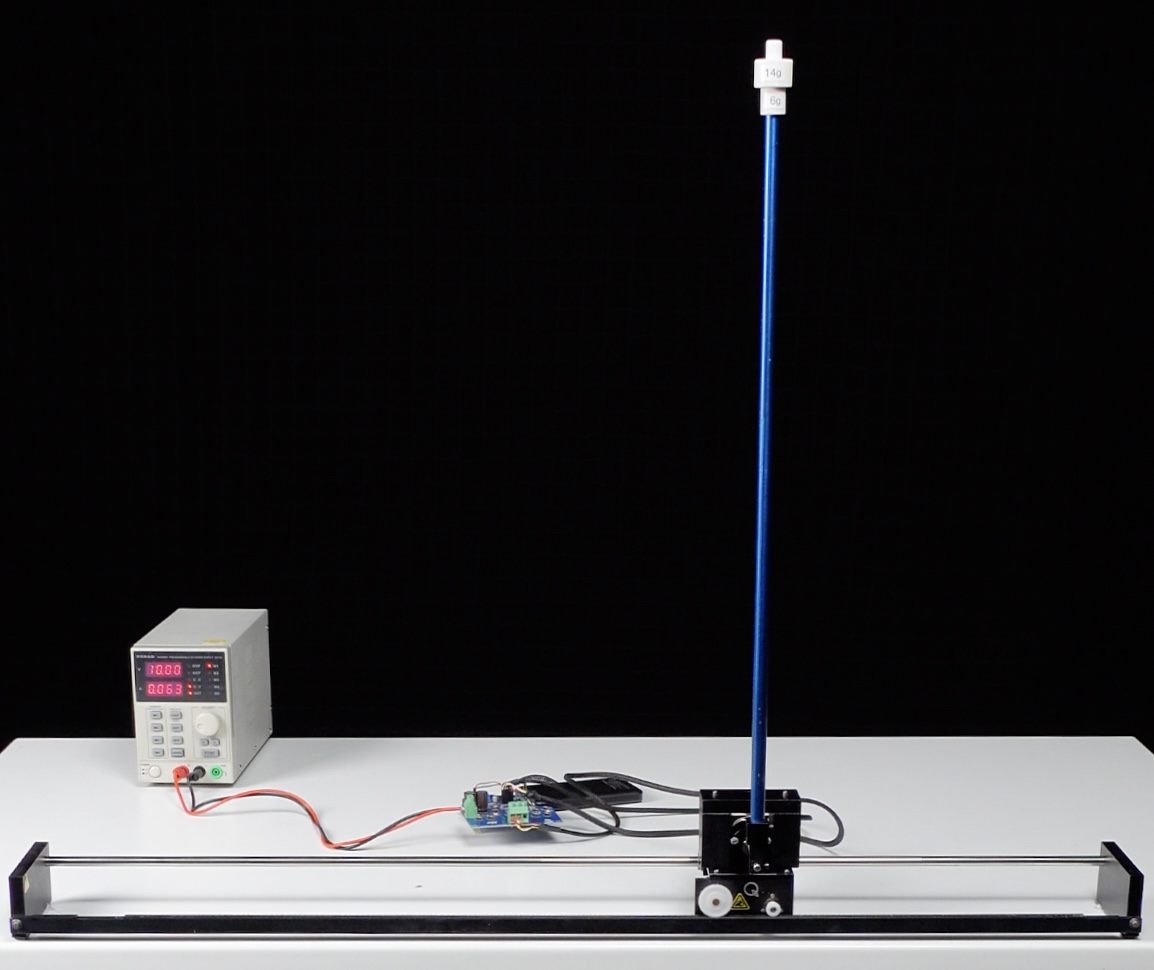}};%
        \node[text opacity=1, anchor=north west, text=white, text centered, text width=1.5cm, inner sep=0.4cm, font=\footnotesize] at (firstpic.north west) {Quanser system};%
        \node (secondpic) [right=0cm of firstpic] {\includegraphics[clip, trim=0cm 0cm 0cm 0cm, height=4.4cm]{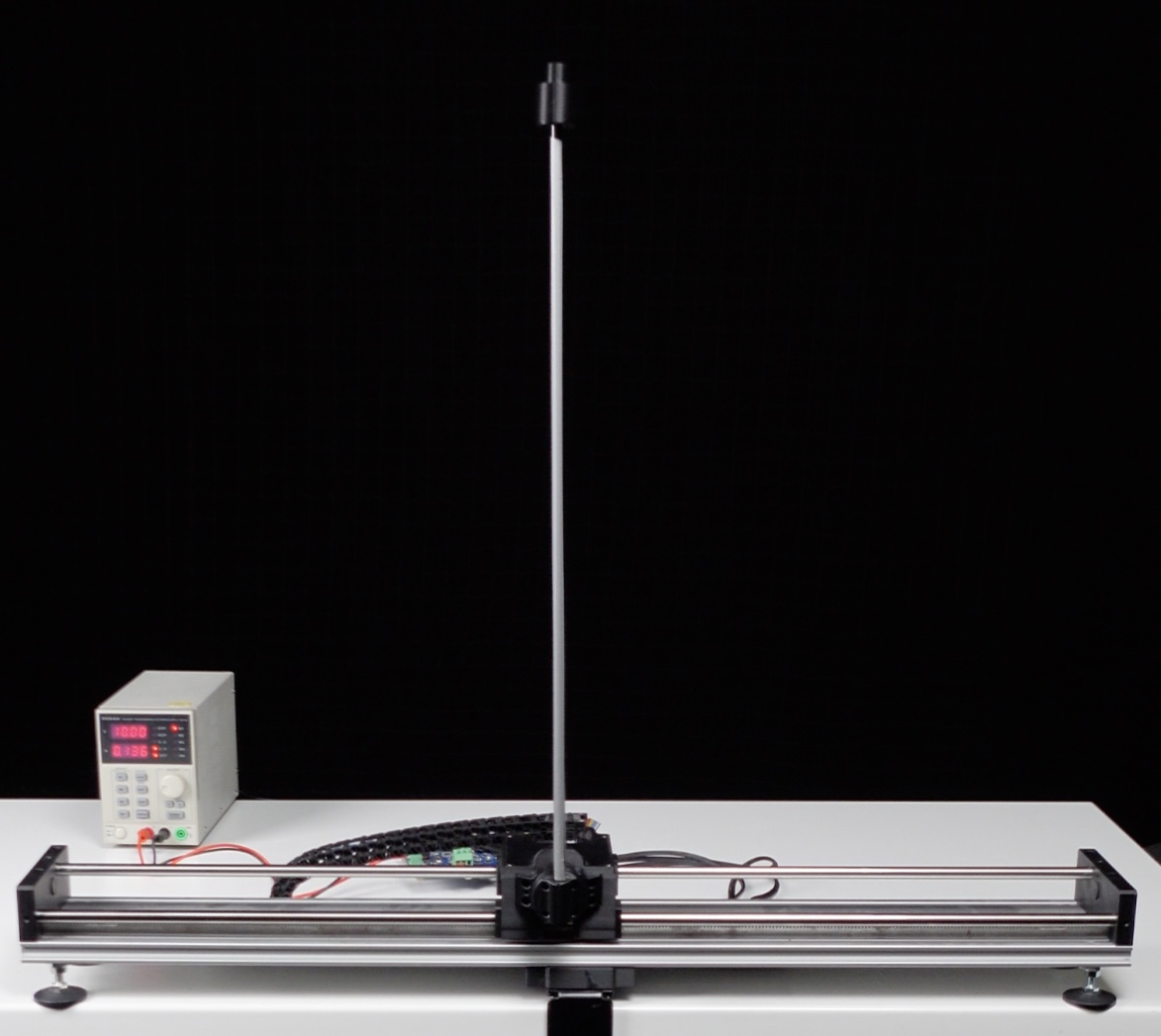}};%
        \node[text opacity=1, anchor=north west, text centered, text=white, text width=1.5cm, inner sep=0.4cm, font=\footnotesize] at (secondpic.north west) {self-made\\pendulum};%
    \end{tikzpicture}  %
    \vspace{-0.2cm}
    \caption{The cartpole inverted pendulum system (left) and the hardware pendulums used. \capt{The Quanser pendulum (center) and the self-made pendulum (right) have significantly different parameters~$\param$ from each other and from $\paramnom$. A video of our experiments is available at~\videolink.}}
    \label{fig:cartpole}
\end{figure}

\subsection{Implementation of the nonlinear MPC with sensitivities}\label{sec:impl-nlmpc-with-sens}
We implement an MPC~(\ref{eqn:implicit-mpc}) with a horizon length of~$N=25$ using collocation and inputs that are step-wise constant for \SI{160}{\milli\second}.
We impose constraints on the cart position~$\pendpos\in[\pendpos_\text{min}, \pendpos_\text{max}]$ and the control~$u\in[u_\text{min},u_\text{max}]$ and choose a stage cost~$\stagecost{\param}{\kappa}{k}=E_\text{kin}-E_\text{pot}+p^2+0.01u^2$ with potential energy~$E_\text{pot}$ and kinetic energy~$E_\text{kin}$ of the system, which is known to work well for pendulum swing-up tasks \citep{magni2002global,boubaker2013inverted}.
We impose a terminal constraint~$\sysstatempc{N}{k}\in\mathcal{X}(N)$ to avoid set-valued MPC solutions (ambiguities) in our dataset that would cause a single feedforward NN to average through different solutions~\citep{carius2020mpc,li2022using}.
To avoid ambiguities from local minima, we reinitialize the NLP solver multiple times with random states and only consider the minimal cost value from all successful runs.
We compute sensitivities for the parameters~$\param=[m_\text{add}, M, C_1, C_2, C_3]^\top$ with nominal values~$\paramnom=\defaultparamvalues$ from direct measurements or manufacturer datasheets. 
We found that the true parameters~$\param$ identified on trajectory data by minimizing a squared loss for the prediction error along an MPC horizon to have significant deviations among several instances of the system,~\ie as depicted in~\fig\ref{fig:cartpole} for the Quanser system $\paramnom+\quansertunedparamvalues$ and for the self-made pendulum $\paramnom+\mpitunedparamvalues$.

The MPC optimization problem is formulated using Pyomo~\citep{bynum2021pyomo} and solved using IPOPT~\citep{wachter2006implementation} with sensitivities through the sIPOPT extension~\citep{pirnay2011sipopt}. The dataset used for the training contains optimal inputs for \num{368000} randomly sampled initial states with $1\%$ of initial states densely sampled around the upright position. Computation of the dataset takes over~$\num{12000}$ CPU core hours\footnote{computed in parallel with Intel Xeon Platinum 8160 ”SkyLake” CPU at 2.1 GHz} due to random reinitializations and high numerical accuracy for sensitivities. 
This shows that without the proposed method, tuning a nominal AMPC by recomputing the dataset would require large amounts of resources and time.

\subsection{Implementation of the parameter-adaptive AMPC}
We use fully-connected feedforward NNs with $50$ neurons per layer, tangent-hyperbolic activations, and $5$ and $8$ layers for~$\ampccontr{\cdot}$ and~$\ampcgrad{\cdot}$, respectively.
While ReLU activations are well suited to represent piecewise affine solutions of linear MPC, tangent-hyperbolic activations yield smooth policies for nonlinear systems which is favored in practical applications~\cite{carius2020mpc,nubert2020safe}.
Training both NNs in Jax~\citep{jax2018github} with a mixture of Minkowski and linear loss takes~$4$ hours\footnote{on an Nvidia RTX4070}.
For hardware experiments, the parameter-adaptive AMPC is implemented on a STM32G474 MCU\footnote{Arm Cortex-M4 core at 170MHz with 96kByte of SRAM and 512kByte flash}.
We use the embedded C++ library builder \textit{modm} for hardware abstraction.
The NN forward pass is implemented as single-precision floating point matrix-vector multiplication with a lookup table for the tangent hyperbolic and takes less than 2ms to evaluate for both NNs.
Due to the symmetry of the problem, we can reset~$\pendang\gets\mathrm{mod}(\pendang + \pi, 2\pi) - \pi$ before evaluating the AMPC in practical experiments.
Hyperparameters in this example were chosen for simplicity and based on engineering intuition, thus further tuning, quantization, or pruning could improve performance.

\subsection{Simulation results: study of parameter variations}\label{sec:results}

\begin{figure}[t]
    \def\figwidth{0.5\linewidth}
    \def\figheight{5cm}
    \begin{tikzpicture}

\definecolor{darkgray176}{RGB}{176,176,176}

\begin{groupplot}[group style={group size=2 by 1}]
\nextgroupplot[
height=\figheight,
label style={font=\footnotesize},
tick align=outside,
tick pos=left,
ticklabel style={font=\footnotesize},
title={\footnotesize\textbf{parameter-adaptive AMPC (ours)}},
width=\figwidth,
x grid style={darkgray176},
xmin=-0.5, xmax=4.5,
xtick style={color=black},
xtick={0,1,2,3,4},
xticklabels={
  \(\displaystyle m_\mathrm{add}\),
  \(\displaystyle C_1\),
  \(\displaystyle M\),
  \(\displaystyle C_2\),
  \(\displaystyle C_3\)
},
y dir=reverse,
y grid style={darkgray176},
ylabel={norm. parameter deviation},
ymin=-0.5, ymax=8.5,
ytick style={color=black},
ytick={0,1,2,3,4,5,6,7,8},
yticklabels={-1.00,-0.75,-0.50,-0.25,0.00,0.25,0.50,0.75,1.00}
]
\addplot graphics [includegraphics cmd=\pgfimage,xmin=-0.5, xmax=4.5, ymin=8.5, ymax=-0.5] {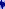};
\draw (axis cs:0,0) node[
  scale=0.5,
  text=white,
  rotate=0.0
]{0.02};
\draw (axis cs:1,0) node[
  scale=0.5,
  text=white,
  rotate=0.0
]{0.995};
\draw (axis cs:2,0) node[
  scale=0.5,
  text=white,
  rotate=0.0
]{0.0};
\draw (axis cs:3,0) node[
  scale=0.5,
  text=white,
  rotate=0.0
]{0.0};
\draw (axis cs:4,0) node[
  scale=0.5,
  text=white,
  rotate=0.0
]{0.055};
\draw (axis cs:0,1) node[
  scale=0.5,
  text=white,
  rotate=0.0
]{0.65};
\draw (axis cs:1,1) node[
  scale=0.5,
  text=white,
  rotate=0.0
]{0.98};
\draw (axis cs:2,1) node[
  scale=0.5,
  text=white,
  rotate=0.0
]{0.3};
\draw (axis cs:3,1) node[
  scale=0.5,
  text=white,
  rotate=0.0
]{0.04};
\draw (axis cs:4,1) node[
  scale=0.5,
  text=white,
  rotate=0.0
]{0.285};
\draw (axis cs:0,2) node[
  scale=0.5,
  text=white,
  rotate=0.0
]{1.0};
\draw (axis cs:1,2) node[
  scale=0.5,
  text=white,
  rotate=0.0
]{1.0};
\draw (axis cs:2,2) node[
  scale=0.5,
  text=white,
  rotate=0.0
]{0.99};
\draw (axis cs:3,2) node[
  scale=0.5,
  text=white,
  rotate=0.0
]{0.945};
\draw (axis cs:4,2) node[
  scale=0.5,
  text=white,
  rotate=0.0
]{0.46};
\draw (axis cs:0,3) node[
  scale=0.5,
  text=white,
  rotate=0.0
]{1.0};
\draw (axis cs:1,3) node[
  scale=0.5,
  text=white,
  rotate=0.0
]{1.0};
\draw (axis cs:2,3) node[
  scale=0.5,
  text=white,
  rotate=0.0
]{0.88};
\draw (axis cs:3,3) node[
  scale=0.5,
  text=white,
  rotate=0.0
]{1.0};
\draw (axis cs:4,3) node[
  scale=0.5,
  text=white,
  rotate=0.0
]{0.985};
\draw (axis cs:0,4) node[
  scale=0.5,
  text=white,
  rotate=0.0
]{1.0};
\draw (axis cs:1,4) node[
  scale=0.5,
  text=white,
  rotate=0.0
]{1.0};
\draw (axis cs:2,4) node[
  scale=0.5,
  text=white,
  rotate=0.0
]{1.0};
\draw (axis cs:3,4) node[
  scale=0.5,
  text=white,
  rotate=0.0
]{1.0};
\draw (axis cs:4,4) node[
  scale=0.5,
  text=white,
  rotate=0.0
]{1.0};
\draw (axis cs:0,5) node[
  scale=0.5,
  text=white,
  rotate=0.0
]{0.37};
\draw (axis cs:1,5) node[
  scale=0.5,
  text=white,
  rotate=0.0
]{0.995};
\draw (axis cs:2,5) node[
  scale=0.5,
  text=white,
  rotate=0.0
]{1.0};
\draw (axis cs:3,5) node[
  scale=0.5,
  text=white,
  rotate=0.0
]{0.815};
\draw (axis cs:4,5) node[
  scale=0.5,
  text=white,
  rotate=0.0
]{0.965};
\draw (axis cs:0,6) node[
  scale=0.5,
  text=white,
  rotate=0.0
]{0.255};
\draw (axis cs:1,6) node[
  scale=0.5,
  text=white,
  rotate=0.0
]{0.885};
\draw (axis cs:2,6) node[
  scale=0.5,
  text=white,
  rotate=0.0
]{0.995};
\draw (axis cs:3,6) node[
  scale=0.5,
  text=white,
  rotate=0.0
]{0.775};
\draw (axis cs:4,6) node[
  scale=0.5,
  text=white,
  rotate=0.0
]{0.345};
\draw (axis cs:0,7) node[
  scale=0.5,
  text=white,
  rotate=0.0
]{0.175};
\draw (axis cs:1,7) node[
  scale=0.5,
  text=white,
  rotate=0.0
]{0.61};
\draw (axis cs:2,7) node[
  scale=0.5,
  text=white,
  rotate=0.0
]{0.99};
\draw (axis cs:3,7) node[
  scale=0.5,
  text=white,
  rotate=0.0
]{0.18};
\draw (axis cs:4,7) node[
  scale=0.5,
  text=white,
  rotate=0.0
]{0.035};
\draw (axis cs:0,8) node[
  scale=0.5,
  text=white,
  rotate=0.0
]{0.08};
\draw (axis cs:1,8) node[
  scale=0.5,
  text=white,
  rotate=0.0
]{0.18};
\draw (axis cs:2,8) node[
  scale=0.5,
  text=white,
  rotate=0.0
]{0.93};
\draw (axis cs:3,8) node[
  scale=0.5,
  text=white,
  rotate=0.0
]{0.01};
\draw (axis cs:4,8) node[
  scale=0.5,
  text=white,
  rotate=0.0
]{0.015};

\nextgroupplot[
height=\figheight,
label style={font=\footnotesize},
scaled y ticks=manual:{}{\pgfmathparse{#1}},
tick align=outside,
tick pos=left,
ticklabel style={font=\footnotesize},
title={\footnotesize nominal AMPC},
width=\figwidth,
x grid style={darkgray176},
xmin=-0.5, xmax=4.5,
xtick style={color=black},
xtick={0,1,2,3,4},
xticklabels={
  \(\displaystyle m_\mathrm{add}\),
  \(\displaystyle C_1\),
  \(\displaystyle M\),
  \(\displaystyle C_2\),
  \(\displaystyle C_3\)
},
y dir=reverse,
y grid style={darkgray176},
ymin=-0.5, ymax=8.5,
ytick style={color=black},
yticklabels={}
]
\addplot graphics [includegraphics cmd=\pgfimage,xmin=-0.5, xmax=4.5, ymin=8.5, ymax=-0.5] {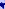};
\draw (axis cs:0,0) node[
  scale=0.5,
  text=white,
  rotate=0.0
]{0.0};
\draw (axis cs:1,0) node[
  scale=0.5,
  text=white,
  rotate=0.0
]{0.0};
\draw (axis cs:2,0) node[
  scale=0.5,
  text=white,
  rotate=0.0
]{0.0};
\draw (axis cs:3,0) node[
  scale=0.5,
  text=white,
  rotate=0.0
]{0.0};
\draw (axis cs:4,0) node[
  scale=0.5,
  text=white,
  rotate=0.0
]{0.145};
\draw (axis cs:0,1) node[
  scale=0.5,
  text=white,
  rotate=0.0
]{0.0};
\draw (axis cs:1,1) node[
  scale=0.5,
  text=white,
  rotate=0.0
]{0.0};
\draw (axis cs:2,1) node[
  scale=0.5,
  text=white,
  rotate=0.0
]{0.0};
\draw (axis cs:3,1) node[
  scale=0.5,
  text=white,
  rotate=0.0
]{0.0};
\draw (axis cs:4,1) node[
  scale=0.5,
  text=white,
  rotate=0.0
]{0.065};
\draw (axis cs:0,2) node[
  scale=0.5,
  text=white,
  rotate=0.0
]{0.03};
\draw (axis cs:1,2) node[
  scale=0.5,
  text=white,
  rotate=0.0
]{0.735};
\draw (axis cs:2,2) node[
  scale=0.5,
  text=white,
  rotate=0.0
]{0.15};
\draw (axis cs:3,2) node[
  scale=0.5,
  text=white,
  rotate=0.0
]{0.01};
\draw (axis cs:4,2) node[
  scale=0.5,
  text=white,
  rotate=0.0
]{0.065};
\draw (axis cs:0,3) node[
  scale=0.5,
  text=white,
  rotate=0.0
]{0.98};
\draw (axis cs:1,3) node[
  scale=0.5,
  text=white,
  rotate=0.0
]{0.895};
\draw (axis cs:2,3) node[
  scale=0.5,
  text=white,
  rotate=0.0
]{0.935};
\draw (axis cs:3,3) node[
  scale=0.5,
  text=white,
  rotate=0.0
]{0.775};
\draw (axis cs:4,3) node[
  scale=0.5,
  text=white,
  rotate=0.0
]{1.0};
\draw (axis cs:0,4) node[
  scale=0.5,
  text=white,
  rotate=0.0
]{1.0};
\draw (axis cs:1,4) node[
  scale=0.5,
  text=white,
  rotate=0.0
]{1.0};
\draw (axis cs:2,4) node[
  scale=0.5,
  text=white,
  rotate=0.0
]{1.0};
\draw (axis cs:3,4) node[
  scale=0.5,
  text=white,
  rotate=0.0
]{1.0};
\draw (axis cs:4,4) node[
  scale=0.5,
  text=white,
  rotate=0.0
]{1.0};
\draw (axis cs:0,5) node[
  scale=0.5,
  text=white,
  rotate=0.0
]{0.99};
\draw (axis cs:1,5) node[
  scale=0.5,
  text=white,
  rotate=0.0
]{0.995};
\draw (axis cs:2,5) node[
  scale=0.5,
  text=white,
  rotate=0.0
]{1.0};
\draw (axis cs:3,5) node[
  scale=0.5,
  text=white,
  rotate=0.0
]{1.0};
\draw (axis cs:4,5) node[
  scale=0.5,
  text=white,
  rotate=0.0
]{0.055};
\draw (axis cs:0,6) node[
  scale=0.5,
  text=white,
  rotate=0.0
]{0.115};
\draw (axis cs:1,6) node[
  scale=0.5,
  text=white,
  rotate=0.0
]{0.96};
\draw (axis cs:2,6) node[
  scale=0.5,
  text=white,
  rotate=0.0
]{0.92};
\draw (axis cs:3,6) node[
  scale=0.5,
  text=white,
  rotate=0.0
]{1.0};
\draw (axis cs:4,6) node[
  scale=0.5,
  text=white,
  rotate=0.0
]{0.055};
\draw (axis cs:0,7) node[
  scale=0.5,
  text=white,
  rotate=0.0
]{0.055};
\draw (axis cs:1,7) node[
  scale=0.5,
  text=white,
  rotate=0.0
]{0.0};
\draw (axis cs:2,7) node[
  scale=0.5,
  text=white,
  rotate=0.0
]{0.76};
\draw (axis cs:3,7) node[
  scale=0.5,
  text=white,
  rotate=0.0
]{1.0};
\draw (axis cs:4,7) node[
  scale=0.5,
  text=white,
  rotate=0.0
]{0.05};
\draw (axis cs:0,8) node[
  scale=0.5,
  text=white,
  rotate=0.0
]{0.045};
\draw (axis cs:1,8) node[
  scale=0.5,
  text=white,
  rotate=0.0
]{0.0};
\draw (axis cs:2,8) node[
  scale=0.5,
  text=white,
  rotate=0.0
]{0.15};
\draw (axis cs:3,8) node[
  scale=0.5,
  text=white,
  rotate=0.0
]{0.67};
\draw (axis cs:4,8) node[
  scale=0.5,
  text=white,
  rotate=0.0
]{0.045};
\end{groupplot}

\end{tikzpicture}%
    \vspace{-0.5cm}
    \caption{Approximate linear predictor effectiveness: \capt{Each cell indicates the fraction of closed-loop simulations from random initial states where the pendulum is stabilized upright, given a deviation in parameters (scaled to $\pm1$ in this plot). The use of approximate linear predictors enables the NN to stabilize the system across a broader parameter range.}}\label{fig:heatmap}
\end{figure}
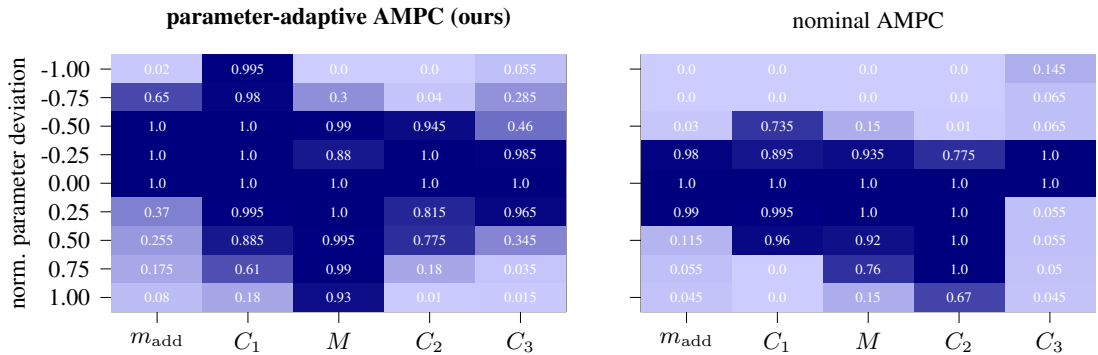

We test parameter-adaptive AMPC in simulation by systematically changing each parameter individually within the bounds~$\param-\paramnom\in\pm\simuparamvalues$ while keeping all other parameters at nominal values.
We evaluate for the same random initial conditions if both AMPCs successfully perform the swing-up and stabilization within constraints.
As shown in~\fig\ref{fig:heatmap}, the parameter-adaptive AMPC (our method) stabilizes the system across a broader range of parameter values compared to the nominal controller.
Interestingly, the nominal controller is able to stabilize the system for a broader range of values in parameter~$C_2$, which we attribute to the linear predictions becoming inaccurate.
However, because one is free to choose the parameters in the AMPC, if the real pendulum dynamics lie in this range, one can just choose to use the nominal parameters to control the pendulum.

\subsection{Hardware results: transfer to system instances without retraining}
We test parameter-adaptive AMPC on two real cartpole systems (see \fig\ref{fig:cartpole}).
A video of our experiments can be found at~\videolink.
The controller with nominal parameters does not perform a swing-up and reaches a limit cycle on both the Quanser and self-made systems.
When using parameters identified from real trajectory data for each system instance (see \sect\ref{sec:impl-nlmpc-with-sens}), the parameter-adaptive AMPC (our method) reliably performs a successful swing-up, as depicted for an exemplary trajectory in~\fig\ref{fig:timeseries}.
In addition, the corrected parameter values for the Quanser pendulum fail on the self-made pendulum where the pendulum starts rotating indefinitely (depicted in orange in~\fig\ref{fig:timeseries}, right) or the cart violates its constraints.
There is a notable residual offset of the cart position in our experiments, which is due to an unmodeled friction effect in the pendulum bearing.
However, as can be seen in our video, the cart does not violate its constraints over a long horizon, but slowly drives back and forth on its rail.

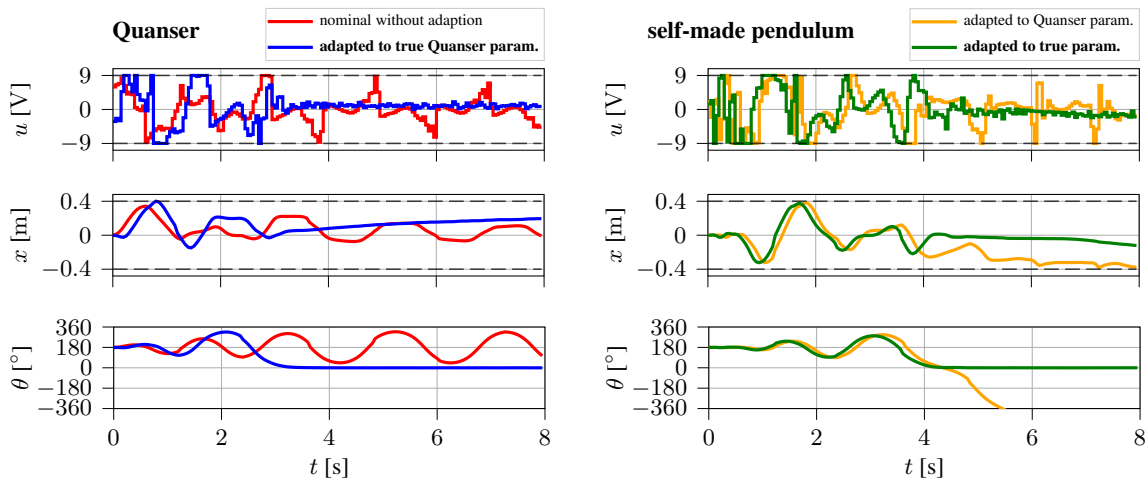
\begin{figure}[t]
    \begin{minipage}[b]{0.48\textwidth}%
        \def\figwidth{\linewidth}%
        \def\figheight{8cm}%
        \begin{tikzpicture}

\definecolor{darkgray176}{RGB}{176,176,176}
\definecolor{lightgray204}{RGB}{204,204,204}

\begin{axis}[
height=\figheight/3 - 2/3*0.0cm,
label style={font=\footnotesize},
legend cell align={left},
legend style={
  fill opacity=0.8,
  draw opacity=1,
  text opacity=1,
  at={(1,1.1)},
  anchor=south east,
  draw=lightgray204
},
legend style={font=\tiny, inner sep=0.1pt, outer sep=0pt},
scaled x ticks=manual:{}{\pgfmathparse{#1}},
tick align=outside,
tick label style={font=\footnotesize},
tick pos=left,
title style={at={(0.1,1)},above,yshift=0pt},
title={\footnotesize\textbf{Quanser}},
width=\figwidth,
x grid style={darkgray176},
xmajorgrids,
xmin=0, xmax=8,
xtick style={color=black},
xticklabels={},
y grid style={darkgray176},
ylabel={\(\displaystyle u\) [\si{\volt}]},
ymajorgrids,
ymin=-10.8, ymax=10.8,
ytick style={color=black},
ytick={-9, 0, 9}
]
\addplot [very thick, red, const plot mark right]
table {%
0 5.95293
0.05 5.95293
0.1 6.53596
0.15 6.87591
0.2 8.36107
0.25 6.35138
0.3 4.07257
0.35 4.56783
0.4 3.91158
0.45 0.519734
0.5 0.0199204
0.55 -0.329388
0.6 0.912043
0.65 -8.60226
0.7 -7.54219
0.75 -6.31142
0.8 -5.66479
0.85 -5.09857
0.9 -4.48647
0.95 -5.44284
1 -5.00658
1.05 -5.00064
1.1 -5.51275
1.15 -2.41262
1.2 -0.482449
1.25 1.35541
1.3 6.34391
1.35 2.95122
1.4 2.12811
1.45 2.0471
1.5 1.26422
1.55 1.46728
1.6 1.59868
1.65 2.37194
1.7 3.35761
1.75 3.13305
1.8 3.39977
1.85 2.09579
1.9 -1.32733
1.95 -3.23558
2 -5.63576
2.05 -5.54561
2.1 -2.29821
2.15 -1.78175
2.2 -1.35903
2.25 -0.782232
2.3 -0.687294
2.35 -1.21027
2.4 -1.7593
2.45 -2.5553
2.5 -2.5231
2.55 -2.31823
2.6 -0.892772
2.65 2.05681
2.7 3.86718
2.75 5.33426
2.8 8.96188
2.85 9
2.9 9
2.95 8.69219
3 -2.71958
3.05 -2.60453
3.1 -1.98801
3.15 -1.68951
3.2 -1.35486
3.25 -0.864382
3.3 -0.413427
3.35 0.0635717
3.4 0.224468
3.45 0.26198
3.5 0.217699
3.55 -1.53548
3.6 -2.53316
3.65 -3.28747
3.7 -2.85971
3.75 -6.67833
3.8 -7.19096
3.85 -8.68037
3.9 -5.84025
3.95 -0.10205
4 1.48882
4.05 1.18593
4.1 0.879553
4.15 0.645628
4.2 0.285304
4.25 -0.0244349
4.3 -0.424401
4.35 -0.563996
4.4 -0.845785
4.45 -0.947261
4.5 -0.63015
4.55 -0.227273
4.6 0.87283
4.65 1.95786
4.7 2.35707
4.75 2.16678
4.8 4.21313
4.85 4.46721
4.9 8.8011
4.95 4.56865
5 -1.78202
5.05 -1.7933
5.1 -1.37427
5.15 -1.10519
5.2 -0.997264
5.25 -0.769647
5.3 -0.328333
5.35 0.0634392
5.4 0.292694
5.45 0.395238
5.5 0.377072
5.55 0.109631
5.6 -0.694706
5.65 -1.69074
5.7 -2.3298
5.75 -2.91018
5.8 -1.95618
5.85 -4.12192
5.9 -5.71891
5.95 -4.40153
6 -7.36
6.05 -0.054849
6.1 1.03626
6.15 0.784748
6.2 0.604603
6.25 0.208414
6.3 -0.0944898
6.35 -0.271866
6.4 -0.625814
6.45 -0.882701
6.5 -0.974256
6.55 -0.750645
6.6 -0.216398
6.65 0.577824
6.7 1.87876
6.75 2.47317
6.8 1.85546
6.85 4.22433
6.9 5.25991
6.95 3.13394
7 6.99734
7.05 0.642667
7.1 -1.53126
7.15 -1.3476
7.2 -1.01369
7.25 -0.81731
7.3 -0.237367
7.35 0.142296
7.4 0.36743
7.45 0.570553
7.5 0.680891
7.55 0.496789
7.6 -0.0648394
7.65 -0.723101
7.7 -2.14746
7.75 -2.55007
7.8 -2.17146
7.85 -4.64418
7.9 -4.17513
7.95 -4.71211
};
\addlegendentry{nominal without adaption}
\addplot [very thick, blue, const plot mark right]
table {%
0 -3.19376
0.05 -3.19376
0.1 -2.35374
0.15 -2.99712
0.2 4.73338
0.25 9
0.3 9
0.35 6.67791
0.4 4.6508
0.45 9
0.5 3.96851
0.55 2.46899
0.6 2.4139
0.65 -0.374527
0.7 4.15308
0.75 9
0.8 -9
0.85 -8.81973
0.9 -9
0.95 -9
1 -9
1.05 -7.20644
1.1 -5.91651
1.15 -5.13269
1.2 -6.17959
1.25 -5.90396
1.3 -5.67264
1.35 -2.18496
1.4 1.48078
1.45 5.77386
1.5 9
1.55 9
1.6 8.7361
1.65 9
1.7 9
1.75 9
1.8 6.22711
1.85 -4.45451
1.9 -4.15894
1.95 -3.18647
2 -2.66411
2.05 -2.34251
2.1 -1.75062
2.15 -0.8743
2.2 -0.0981351
2.25 0.657988
2.3 1.25248
2.35 1.6434
2.4 1.11916
2.45 -1.74191
2.5 -5.11051
2.55 -6.90394
2.6 -6.63262
2.65 -3.4922
2.7 -1.66302
2.75 -9
2.8 -2.51338
2.85 -2.22973
2.9 5.70515
2.95 6.44873
3 1.18359
3.05 0.379357
3.1 3.23731
3.15 3.36351
3.2 -0.223383
3.25 -0.984683
3.3 0.155689
3.35 1.27185
3.4 0.754354
3.45 0.379765
3.5 0.493034
3.55 1.59187
3.6 0.513462
3.65 1.9848
3.7 0.967195
3.75 0.587319
3.8 1.675
3.85 1.43187
3.9 0.442626
3.95 1.67626
4 0.605096
4.05 1.29737
4.1 1.36903
4.15 0.863847
4.2 0.72887
4.25 1.6495
4.3 1.07301
4.35 0.880906
4.4 1.28071
4.45 1.18635
4.5 0.678726
4.55 1.00995
4.6 1.85966
4.65 1.11678
4.7 0.494657
4.75 0.829909
4.8 0.838926
4.85 0.983957
4.9 0.85375
4.95 0.508978
5 1.25946
5.05 1.14733
5.1 0.753636
5.15 1.52338
5.2 1.04021
5.25 0.75341
5.3 1.15261
5.35 0.883203
5.4 0.857077
5.45 1.3951
5.5 1.48743
5.55 1.47666
5.6 -0.0142027
5.65 1.07672
5.7 0.383643
5.75 0.626268
5.8 1.50857
5.85 1.16646
5.9 0.446616
5.95 0.356123
6 1.03858
6.05 0.900725
6.1 0.778725
6.15 1.18129
6.2 1.03714
6.25 0.972012
6.3 0.737504
6.35 2.03135
6.4 1.32753
6.45 0.173455
6.5 0.498718
6.55 0.801564
6.6 1.72625
6.65 0.833378
6.7 0.613607
6.75 0.85229
6.8 1.11535
6.85 1.86653
6.9 1.18175
6.95 0.529979
7 0.288519
7.05 1.52317
7.1 1.46508
7.15 0.169049
7.2 0.727561
7.25 0.926288
7.3 0.948768
7.35 1.0141
7.4 1.13381
7.45 1.31963
7.5 0.961402
7.55 1.03764
7.6 1.66369
7.65 1.47239
7.7 0.173971
7.75 0.665205
7.8 1.02733
7.85 1.25616
7.9 0.745633
7.95 0.821177
};
\addlegendentry{\textbf{adapted to true Quanser param.}}
\path [draw=black, draw opacity=0.7, semithick, dash pattern=on 5.55pt off 2.4pt]
(axis cs:0,9)
--(axis cs:7.95,9);

\path [draw=black, draw opacity=0.7, semithick, dash pattern=on 5.55pt off 2.4pt]
(axis cs:0,-9)
--(axis cs:7.95,-9);

\end{axis}

\end{tikzpicture}%
        \begin{tikzpicture}

\definecolor{darkgray176}{RGB}{176,176,176}

\begin{axis}[
height=\figheight/3 - 2/3*0.0cm,
label style={font=\footnotesize},
scaled x ticks=manual:{}{\pgfmathparse{#1}},
tick align=outside,
tick label style={font=\footnotesize},
tick pos=left,
width=\figwidth,
x grid style={darkgray176},
xmajorgrids,
xmin=0, xmax=8,
xtick style={color=black},
xticklabels={},
y grid style={darkgray176},
ylabel={\(\displaystyle x\) [\si{\meter}]},
ymajorgrids,
ymin=-0.48, ymax=0.48,
ytick style={color=black},
ytick={-0.4, 0, 0.4}
]
\addplot [very thick, red]
table {%
0 0
0.05 0.00727966
0.1 0.0272987
0.15 0.0584648
0.2 0.0990034
0.25 0.146003
0.3 0.191387
0.35 0.233017
0.4 0.271759
0.45 0.30297
0.5 0.323581
0.55 0.335456
0.6 0.34244
0.65 0.335888
0.7 0.311274
0.75 0.277378
0.8 0.241616
0.85 0.20474
0.9 0.170117
0.95 0.135174
1 0.0995038
1.05 0.063697
1.1 0.0271167
1.15 -0.00732516
1.2 -0.0330542
1.25 -0.0466808
1.3 -0.0293689
1.35 -0.00762089
1.4 0.00996403
1.45 0.0224532
1.5 0.0307793
1.55 0.0360116
1.6 0.0403111
1.65 0.0454069
1.7 0.0539377
1.75 0.0654259
1.8 0.0790525
1.85 0.0919284
1.9 0.0990034
1.95 0.0952953
2 0.0798715
2.05 0.0546884
2.1 0.0297101
2.15 0.0124892
2.2 0.00216115
2.25 -0.00279812
2.3 -0.00395831
2.35 -0.00407206
2.4 -0.00520951
2.45 -0.0169252
2.5 -0.0269802
2.55 -0.036421
2.6 -0.042313
2.65 -0.0406978
2.7 -0.0296874
2.75 -0.00907682
2.8 0.0221347
2.85 0.0618316
2.9 0.106147
2.95 0.153874
3 0.192251
3.05 0.212043
3.1 0.220096
3.15 0.221961
3.2 0.221484
3.25 0.221438
3.3 0.221666
3.35 0.222098
3.4 0.222098
3.45 0.221688
3.5 0.219618
3.55 0.213521
3.6 0.200236
3.65 0.151599
3.7 0.126598
3.75 0.103576
3.8 0.0785976
3.85 0.0481595
3.9 0.0138086
3.95 -0.0150598
4 -0.0338732
4.05 -0.0448836
4.1 -0.0519131
4.15 -0.0571453
4.2 -0.0613311
4.25 -0.0649255
4.3 -0.0681558
4.35 -0.0707947
4.4 -0.0729331
4.45 -0.0742298
4.5 -0.0736611
4.55 -0.0700212
4.6 -0.0618544
4.65 -0.0465671
4.7 -0.0255926
4.75 -0.00368533
4.8 0.0152418
4.85 0.0462713
4.9 0.0670866
4.95 0.0940441
5 0.11561
5.05 0.12728
5.1 0.132604
5.15 0.135106
5.2 0.136448
5.25 0.137586
5.3 0.138859
5.35 0.140156
5.4 0.14118
5.45 0.141794
5.5 0.141498
5.55 0.139724
5.6 0.134765
5.65 0.123891
5.7 0.105737
5.75 0.0817369
5.8 0.059352
5.85 0.0420855
5.9 0.0259793
5.95 0.0103963
6 -0.00875834
6.05 -0.0394694
6.1 -0.0457026
6.15 -0.0498657
6.2 -0.0528913
6.25 -0.0556439
6.3 -0.0580325
6.35 -0.0603984
6.4 -0.0626733
6.45 -0.0648117
6.5 -0.0662676
6.55 -0.0659719
6.6 -0.0628781
6.65 -0.0551889
6.7 -0.0411301
6.75 -0.0200873
6.8 0.00179717
6.85 0.0207925
6.9 0.0382637
6.95 0.0525728
7 0.0679056
7.05 0.0843531
7.1 0.0942488
7.15 0.0983664
7.2 0.0999588
7.25 0.102234
7.3 0.10419
7.35 0.10642
7.4 0.108581
7.45 0.110446
7.5 0.111424
7.55 0.110856
7.6 0.107944
7.65 0.101233
7.7 0.0876971
7.75 0.0669046
7.8 0.0446562
7.85 0.0241366
7.9 0.00675643
7.95 -0.00775739
};
\addplot [very thick, blue]
table {%
0 0
0.05 -0.00323035
0.1 -0.0108057
0.15 -0.0211565
0.2 -0.0244779
0.25 -0.00875834
0.3 0.0256835
0.35 0.0705899
0.4 0.116588
0.45 0.165612
0.5 0.216297
0.55 0.260362
0.6 0.297784
0.65 0.325446
0.7 0.3484
0.75 0.379202
0.8 0.400427
0.85 0.391737
0.9 0.362891
0.95 0.321147
1 0.270985
1.05 0.218322
1.1 0.16725
1.15 0.119341
1.2 0.0247736
1.25 -0.0237499
1.3 -0.0706127
1.35 -0.112607
1.4 -0.141953
1.45 -0.14921
1.5 -0.131898
1.55 -0.096228
1.6 -0.0517993
1.65 -0.00493652
1.7 0.0441557
1.75 0.0972744
1.8 0.151076
1.85 0.191728
1.9 0.210132
1.95 0.214249
2 0.211588
2.05 0.206515
2.1 0.201146
2.15 0.197415
2.2 0.196483
2.25 0.197029
2.3 0.198098
2.35 0.199645
2.4 0.196164
2.45 0.186337
2.5 0.165157
2.55 0.130897
2.6 0.0919512
2.65 0.059261
2.7 0.0390599
2.75 0.0193138
2.8 -0.00350334
2.85 -0.0229309
2.9 -0.0309386
2.95 -0.022021
3 -0.00712042
3.05 0.00391282
3.1 0.0142408
3.15 0.0278447
3.2 0.0402884
3.25 0.0474088
3.3 0.0509121
3.35 0.0532098
3.4 0.0549842
3.45 0.0557576
3.5 0.0557576
3.55 0.0566676
3.6 0.0586013
3.65 0.0607169
3.7 0.0639245
3.75 0.0658809
3.8 0.0680648
3.85 0.0713634
3.9 0.074389
3.95 0.0772781
4 0.0803947
4.05 0.0827606
4.1 0.085627
4.15 0.0881749
4.2 0.089972
4.25 0.0921559
4.3 0.0951588
4.35 0.0975702
4.4 0.0999361
4.45 0.10262
4.5 0.1046
4.55 0.105692
4.6 0.108012
4.65 0.111493
4.7 0.114177
4.75 0.115815
4.8 0.118886
4.85 0.121843
4.9 0.123936
4.95 0.124937
5 0.125824
5.05 0.127144
5.1 0.128349
5.15 0.130033
5.2 0.132376
5.25 0.134128
5.3 0.135493
5.35 0.136767
5.4 0.137654
5.45 0.138746
5.5 0.141089
5.55 0.144342
5.6 0.147049
5.65 0.14871
5.7 0.149779
5.75 0.149938
5.8 0.150894
5.85 0.152805
5.9 0.15417
5.95 0.154329
6 0.154374
6.05 0.155284
6.1 0.156672
6.15 0.157764
6.2 0.158742
6.25 0.159174
6.3 0.159197
6.35 0.160744
6.4 0.163792
6.45 0.166067
6.5 0.166704
6.55 0.166727
6.6 0.16791
6.65 0.169707
6.7 0.170549
6.75 0.17064
6.8 0.170731
6.85 0.172232
6.9 0.174848
6.95 0.176782
7 0.177351
7.05 0.178238
7.1 0.180536
7.15 0.182469
7.2 0.18297
7.25 0.18297
7.3 0.18297
7.35 0.182992
7.4 0.183174
7.45 0.18388
7.5 0.184562
7.55 0.184972
7.6 0.18645
7.65 0.189271
7.7 0.191751
7.75 0.192729
7.8 0.193207
7.85 0.194071
7.9 0.194867
7.95 0.195118
};
\path [draw=black, draw opacity=0.7, semithick, dash pattern=on 5.55pt off 2.4pt]
(axis cs:0,0.4)
--(axis cs:7.95,0.4);

\path [draw=black, draw opacity=0.7, semithick, dash pattern=on 5.55pt off 2.4pt]
(axis cs:0,-0.4)
--(axis cs:7.95,-0.4);

\end{axis}

\end{tikzpicture}%
        \begin{tikzpicture}

\definecolor{darkgray176}{RGB}{176,176,176}

\begin{axis}[
height=\figheight/3 - 2/3*0.0cm,
label style={font=\footnotesize},
legend style={font=\tiny, inner sep=0.1pt, outer sep=0pt},
tick align=outside,
tick label style={font=\footnotesize},
tick pos=left,
width=\figwidth,
x grid style={darkgray176},
xlabel={\(\displaystyle t\) [s]},
xmajorgrids,
xmin=0, xmax=8,
xtick style={color=black},
y grid style={darkgray176},
ylabel={\(\displaystyle \theta\) [\si{\degree}]},
ymajorgrids,
ymin=-360, ymax=360,
ytick style={color=black},
ytick={-360, -180, 0, 180, 360}
]
\addplot [very thick, red]
table {%
0 180.088083460954
0.05 181.142325803995
0.1 183.603752491877
0.15 187.207084065455
0.2 191.601670354108
0.25 196.083919185417
0.3 199.687250758995
0.35 201.796881360666
0.4 202.499900575292
0.45 201.181524688696
0.5 197.578193115118
0.55 192.041128982974
0.6 185.273351506888
0.65 176.748312473137
0.7 166.025407337264
0.75 154.951279072975
0.8 145.019628652177
0.85 136.669914703736
0.9 130.693391942726
0.95 126.914162326043
1 125.507550938997
1.05 126.474703697178
1.1 129.902137227651
1.15 135.87865998866
1.2 144.580170023312
1.25 170.771216754332
1.3 187.822440737425
1.35 204.785429220069
1.4 220.078245729905
1.45 232.99787105231
1.5 243.017756973558
1.55 250.312655621164
1.6 254.707241909817
1.65 256.200942881723
1.7 254.882566995127
1.75 250.663878749579
1.8 243.808438730839
1.85 234.052686353146
1.9 221.484284159156
1.95 206.015569606214
2 188.173663865841
2.05 168.662159110456
2.1 149.677775526591
2.15 133.330143715918
2.2 119.619263678438
2.25 108.720460499459
2.3 101.074338723438
2.35 96.3279563485745
2.4 94.6583573335633
2.45 100.283084008363
2.5 107.490320113313
2.55 117.597868577216
2.6 130.429831356966
2.65 146.338004538773
2.7 164.531133407562
2.75 184.657994834918
2.8 205.927907063559
2.85 226.581889662435
2.9 245.65450874675
2.95 262.002140557423
3 275.185326465588
3.05 285.644671015701
3.1 293.466690834827
3.15 298.916092424316
3.2 301.816404783269
3.25 302.079965369029
3.3 299.706774181597
3.35 294.697404178768
3.4 286.963046902297
3.45 276.416039809529
3.5 262.881057815154
3.55 246.357527961376
3.6 226.933685748646
3.65 180.351644046715
3.7 156.709113588436
3.75 134.472621559409
3.8 114.433422674709
3.85 96.5915169343347
3.9 81.8263945538134
3.95 70.0486741171042
4 60.4688197825168
4.05 53.1738065433522
4.1 47.9883093143002
4.15 44.9120989122428
4.2 44.0331816545121
4.25 45.175774089562
4.3 48.427710647386
4.35 53.8769403495367
4.4 61.7871956691128
4.45 72.334202761881
4.5 85.6056241704963
4.55 102.040918523824
4.6 121.113537608139
4.65 143.349456679371
4.7 167.519681266965
4.75 192.392352111389
4.8 216.299016113222
4.85 257.079860139454
4.9 273.427491950127
4.95 286.523588273432
5 296.806461822645
5.05 305.156175771086
5.1 311.396259117856
5.15 315.527284820749
5.2 317.724577965076
5.25 317.900476008181
5.3 315.966743449615
5.35 311.923953247171
5.4 305.595634399952
5.45 296.894697323095
5.5 285.732333558356
5.55 271.58199489201
5.6 254.619006409367
5.65 234.316246938906
5.7 211.552633738359
5.75 186.415829350379
5.8 161.454923005505
5.85 138.251851176092
5.9 116.894276404795
5.95 98.4375869502462
6 83.0565349399593
6.05 60.732380368277
6.1 53.0859148175791
6.15 47.6367424112079
6.2 44.3847485576044
6.25 43.4179968698801
6.3 44.3847485576044
6.35 47.5488506854348
6.4 52.9101313660329
6.45 60.820042910932
6.5 71.10348941794
6.55 84.2872482839003
6.6 100.546644594123
6.65 119.619263678438
6.7 141.4162970786
6.75 165.937744794609
6.8 190.898651139483
6.85 214.717079640866
6.9 237.041234212549
6.95 256.376840924828
7 272.636810192846
7.05 285.820569058806
7.1 296.542901236885
7.15 304.804952642671
7.2 310.781475403681
7.25 315.966743449615
7.3 315.263724234989
7.35 312.451074418692
7.4 307.353468915413
7.45 299.882672224702
7.5 289.951021803904
7.55 277.119059024154
7.6 261.386783885452
7.65 242.929521473108
7.7 221.308386116051
7.75 196.962836443147
7.8 172.177828141378
7.85 147.91994101113
7.9 125.595786439447
7.95 105.908383640957
};
\addplot [very thick, blue]
table {%
0 179.999847960504
0.05 179.648624832089
0.1 178.681472073908
0.15 177.538994230417
0.2 177.363096187312
0.25 179.736287374744
0.3 184.306771706503
0.35 189.931498381302
0.4 195.117339385031
0.45 200.12670938786
0.5 204.081837047648
0.55 206.103805106665
0.6 206.015569606214
0.65 203.730613919233
0.7 199.423690173234
0.75 195.293237428136
0.8 189.052581123571
0.85 178.857370117013
0.9 166.201305380369
0.95 152.753985928648
1 139.658462563138
1.05 128.320200755294
1.1 119.619263678438
1.15 113.642740917428
1.2 110.21473442916
1.25 112.763823659697
1.3 118.037327206081
1.35 126.386468196728
1.4 138.251851176092
1.45 153.808801229484
1.5 173.583866570629
1.55 195.380899970791
1.6 217.880952585579
1.65 238.359610099145
1.7 256.552738967934
1.75 271.84555547777
1.8 284.150397086
1.85 293.906149463692
1.9 301.992302826374
1.95 308.320048715799
2 312.890533047557
2.05 315.439622278094
2.1 315.87908090696
2.15 314.033010891048
2.2 309.726660102845
2.25 303.04711812721
2.3 293.818486921037
2.35 281.777205898568
2.4 248.994279734568
2.45 228.076163592137
2.5 204.170072548098
2.55 177.627229730868
2.6 150.820253370082
2.65 125.947009567862
2.7 104.677670297016
2.75 86.3086433851218
2.8 71.015826875285
2.85 58.6233227244004
2.9 48.1640354700668
2.95 38.2324423450486
3 29.6191678108469
3.05 23.1152373994194
3.1 18.105466326134
3.15 13.535153881714
3.2 9.75586696925155
3.25 7.1191151960595
3.3 5.44921824299482
3.35 4.21875127090572
3.4 3.25195374655787
3.45 2.54882566995127
3.5 2.19726449643695
3.55 1.75781159714956
3.6 1.31835869786217
3.65 1.05468925012092
3.7 0.703122347028644
3.75 0.615236350833527
3.8 0.527343479144868
3.85 0.263671739572434
3.9 0.175781159714956
3.95 0
4 -0.087890579857478
4.05 -0.087890579857478
4.1 -0.175781159714956
4.15 -0.263671739572434
4.2 -0.175781159714956
4.25 -0.175781159714956
4.3 -0.263671739572434
4.35 -0.263671739572434
4.4 -0.263671739572434
4.45 -0.351562319429912
4.5 -0.351562319429912
4.55 -0.175781159714956
4.6 -0.175781159714956
4.65 -0.351562319429912
4.7 -0.43945289928739
4.75 -0.43945289928739
4.8 -0.263671739572434
4.85 -0.351562319429912
4.9 -0.43945289928739
4.95 -0.351562319429912
5 -0.263671739572434
5.05 -0.263671739572434
5.1 -0.175781159714956
5.15 -0.175781159714956
5.2 -0.263671739572434
5.25 -0.263671739572434
5.3 -0.263671739572434
5.35 -0.263671739572434
5.4 -0.175781159714956
5.45 -0.087890579857478
5.5 -0.087890579857478
5.55 -0.351562319429912
5.6 -0.43945289928739
5.65 -0.527343479144868
5.7 -0.527343479144868
5.75 -0.351562319429912
5.8 -0.263671739572434
5.85 -0.351562319429912
5.9 -0.43945289928739
5.95 -0.351562319429912
6 -0.263671739572434
6.05 -0.087890579857478
6.1 -0.087890579857478
6.15 -0.087890579857478
6.2 -0.087890579857478
6.25 -0.087890579857478
6.3 0.087890579857478
6.35 0.087890579857478
6.4 -0.175781159714956
6.45 -0.351562319429912
6.5 -0.351562319429912
6.55 -0.175781159714956
6.6 -0.175781159714956
6.65 -0.263671739572434
6.7 -0.263671739572434
6.75 -0.175781159714956
6.8 0
6.85 0
6.9 -0.175781159714956
6.95 -0.351562319429912
7 -0.263671739572434
7.05 -0.175781159714956
7.1 -0.351562319429912
7.15 -0.43945289928739
7.2 -0.351562319429912
7.25 -0.263671739572434
7.3 -0.175781159714956
7.35 -0.087890579857478
7.4 0
7.45 0
7.5 0
7.55 0.087890579857478
7.6 0.087890579857478
7.65 -0.175781159714956
7.7 -0.351562319429912
7.75 -0.351562319429912
7.8 -0.263671739572434
7.85 -0.263671739572434
7.9 -0.263671739572434
7.95 -0.175781159714956
};
\end{axis}

\end{tikzpicture}%
    \end{minipage}%
    \hspace*{0.5cm}
    \begin{minipage}[b]{0.48\textwidth}%
        \def\figwidth{\linewidth}%
        \def\figheight{8cm}%
        \begin{tikzpicture}

\definecolor{darkgray176}{RGB}{176,176,176}
\definecolor{green}{RGB}{0,128,0}
\definecolor{lightgray204}{RGB}{204,204,204}
\definecolor{orange}{RGB}{255,165,0}

\begin{axis}[
height=\figheight/3 - 2/3*0.0cm,
label style={font=\footnotesize},
legend cell align={left},
legend style={
  fill opacity=0.8,
  draw opacity=1,
  text opacity=1,
  at={(1,1.1)},
  anchor=south east,
  draw=lightgray204
},
legend style={font=\tiny, inner sep=0.1pt, outer sep=0pt},
scaled x ticks=manual:{}{\pgfmathparse{#1}},
tick align=outside,
tick label style={font=\footnotesize},
tick pos=left,
title style={at={(0.1,1)},above,yshift=0pt},
title={\footnotesize\textbf{self-made pendulum}},
width=\figwidth,
x grid style={darkgray176},
xmajorgrids,
xmin=0, xmax=8,
xtick style={color=black},
xticklabels={},
y grid style={darkgray176},
ylabel={\(\displaystyle u\) [\si{\volt}]},
ymajorgrids,
ymin=-10.8, ymax=10.8,
ytick style={color=black},
ytick={-9, 0, 9}
]
\addplot [very thick, orange, const plot mark right]
table {%
0 1.92344
0.05 1.92344
0.1 -2.11672
0.15 1.4162
0.2 2.69039
0.25 -8.05706
0.3 -2.6035
0.35 6.85217
0.4 9
0.45 -2.67663
0.5 -9
0.55 3.98724
0.6 4.51334
0.65 -9
0.7 -9
0.75 -4.84778
0.8 -5.9012
0.85 -6.48748
0.9 -9
0.95 -8.79402
1 2.18054
1.05 7.30088
1.1 4.19861
1.15 7.72859
1.2 9
1.25 8.04603
1.3 9
1.35 9
1.4 9
1.45 8.66002
1.5 5.72179
1.55 5.90515
1.6 3.02063
1.65 -6.71
1.7 -0.790519
1.75 8.59075
1.8 -3.92484
1.85 -9
1.9 -9
1.95 -8.9251
2 -9
2.05 -6.28546
2.1 -3.4092
2.15 -2.07914
2.2 -0.776287
2.25 -1.05659
2.3 -1.48151
2.35 -2.57096
2.4 -3.82529
2.45 -4.86012
2.5 -3.20254
2.55 -1.51233
2.6 0.559468
2.65 4.73194
2.7 9
2.75 9
2.8 6.265
2.85 4.40984
2.9 1.99358
2.95 -0.0612699
3 -1.06887
3.05 -0.999618
3.1 -0.814228
3.15 -0.303617
3.2 0.399199
3.25 1.17671
3.3 2.10844
3.35 3.18052
3.4 3.86487
3.45 4.34227
3.5 2.30692
3.55 0.958989
3.6 -1.02842
3.65 -4.95106
3.7 -9
3.75 -9
3.8 -9
3.85 -9
3.9 -0.299148
3.95 1.14667
4 -1.68607
4.05 -1.83345
4.1 2.90102
4.15 3.2441
4.2 2.03701
4.25 2.59178
4.3 2.61122
4.35 2.28804
4.4 1.92982
4.45 1.80694
4.5 1.91621
4.55 2.03796
4.6 2.30471
4.65 2.60513
4.7 2.81181
4.75 2.74775
4.8 2.25036
4.85 1.17041
4.9 -3.84675
4.95 -5.43106
5 -5.16342
5.05 -6.17798
5.1 -9
5.15 -6.16214
5.2 1.39089
5.25 3.57153
5.3 0.434876
5.35 0.0602773
5.4 1.18572
5.45 1.40278
5.5 1.41223
5.55 1.56001
5.6 1.75721
5.65 1.74199
5.7 1.61137
5.75 0.902303
5.8 0.776991
5.85 0.217275
5.9 0.434723
5.95 0.114682
6 0.464281
6.05 -3.90937
6.1 -8.99697
6.15 6.96148
6.2 7.0181
6.25 -2.25339
6.3 -0.482798
6.35 2.30411
6.4 0.92695
6.45 0.21161
6.5 -0.0188995
6.55 -0.0581391
6.6 -0.0193504
6.65 0.0446222
6.7 0.334965
6.75 0.894018
6.8 1.56543
6.85 2.05803
6.9 1.75002
6.95 1.08316
7 0.90618
7.05 0.902572
7.1 1.14412
7.15 -4.5419
7.2 -9
7.25 -5.18262
7.3 6.35967
7.35 -1.36478
7.4 0.0737769
7.45 2.12028
7.5 0.300743
7.55 -0.370319
7.6 -2.46096
7.65 -2.8209
7.7 -1.66777
7.75 -1.55828
7.8 -1.45132
7.85 -1.30562
7.9 -1.31279
7.95 -1.66577
};
\addlegendentry{adapted to Quanser param.}
\addplot [very thick, green, const plot mark right]
table {%
0 1.46737
0.05 1.46737
0.1 2.37063
0.15 -9
0.2 -1.60075
0.25 8.28918
0.3 9
0.35 -2.28372
0.4 -8.59494
0.45 2.00337
0.5 3.81213
0.55 -9
0.6 -9
0.65 -5.17237
0.7 -9
0.75 -9
0.8 -9
0.85 0.0399145
0.9 3.1367
0.95 5.65604
1 5.77212
1.05 9
1.1 9
1.15 9
1.2 9
1.25 9
1.3 7.88481
1.35 7.61869
1.4 5.15768
1.45 6.55389
1.5 -0.542721
1.55 -6.29635
1.6 2.2077
1.65 9
1.7 -8.69792
1.75 -9
1.8 -8.35202
1.85 -9
1.9 -6.71664
1.95 -3.43997
2 -2.32715
2.05 -1.683
2.1 -1.3028
2.15 -1.5364
2.2 -3.02627
2.25 -4.47598
2.3 -4.87704
2.35 -5.89778
2.4 -3.73067
2.45 -0.182053
2.5 3.3976
2.55 7.81661
2.6 9
2.65 5.00895
2.7 3.02557
2.75 0.377088
2.8 0.167447
2.85 -0.0199346
2.9 -0.17677
2.95 0.290896
3 0.959401
3.05 2.08339
3.1 3.06623
3.15 3.75679
3.2 4.63265
3.25 5.09695
3.3 2.70072
3.35 1.25894
3.4 -0.606146
3.45 -4.37631
3.5 -7.25924
3.55 -8.55829
3.6 -8.72348
3.65 -9
3.7 -5.74822
3.75 3.19665
3.8 7.52731
3.85 9
3.9 7.32479
3.95 5.00089
4 4.99972
4.05 4.56914
4.1 1.45298
4.15 -2.23825
4.2 -2.11006
4.25 -0.692493
4.3 -1.42524
4.35 -1.39116
4.4 -2.04596
4.45 -5.17091
4.5 -0.0648425
4.55 1.48386
4.6 0.858399
4.65 0.0181138
4.7 -0.104718
4.75 3.0369
4.8 -1.03181
4.85 -0.497333
4.9 -0.147613
4.95 -1.47394
5 -1.70634
5.05 0.0828628
5.1 -0.160615
5.15 -1.24952
5.2 -1.37053
5.25 -0.574397
5.3 -0.745602
5.35 -1.85514
5.4 -0.670364
5.45 -0.118623
5.5 -2.51724
5.55 -0.295022
5.6 0.155728
5.65 -0.816348
5.7 -1.20439
5.75 0.098385
5.8 -1.14592
5.85 -0.0942466
5.9 -1.20676
5.95 -0.177779
6 -0.397971
6.05 -0.475617
6.1 -0.474153
6.15 -0.572788
6.2 -0.545317
6.25 -1.4013
6.3 -0.684322
6.35 -0.694588
6.4 -0.63459
6.45 -1.45438
6.5 -0.721148
6.55 -0.793546
6.6 -1.54229
6.65 -0.822327
6.7 -1.80349
6.75 -1.10531
6.8 -0.6959
6.85 -2.12636
6.9 -1.18858
6.95 -1.241
7 -1.66245
7.05 -1.77949
7.1 -0.568366
7.15 -2.2014
7.2 -1.56797
7.25 -1.62406
7.3 -1.92975
7.35 -1.30118
7.4 -0.789423
7.45 -2.03058
7.5 -1.71342
7.55 -1.01468
7.6 -1.87142
7.65 -1.35903
7.7 -1.121
7.75 -1.32833
7.8 -1.69808
7.85 -2.0358
7.9 -0.382229
7.95 -1.98698
};
\addlegendentry{\textbf{adapted to true param.}}
\path [draw=black, draw opacity=0.7, semithick, dash pattern=on 5.55pt off 2.4pt]
(axis cs:0,9)
--(axis cs:7.95,9);

\path [draw=black, draw opacity=0.7, semithick, dash pattern=on 5.55pt off 2.4pt]
(axis cs:0,-9)
--(axis cs:7.95,-9);

\end{axis}

\end{tikzpicture}%
        \begin{tikzpicture}

\definecolor{darkgray176}{RGB}{176,176,176}
\definecolor{green}{RGB}{0,128,0}
\definecolor{orange}{RGB}{255,165,0}

\begin{axis}[
height=\figheight/3 - 2/3*0.0cm,
label style={font=\footnotesize},
scaled x ticks=manual:{}{\pgfmathparse{#1}},
tick align=outside,
tick label style={font=\footnotesize},
tick pos=left,
width=\figwidth,
x grid style={darkgray176},
xmajorgrids,
xmin=0, xmax=8,
xtick style={color=black},
xticklabels={},
y grid style={darkgray176},
ylabel={\(\displaystyle x\) [\si{\meter}]},
ymajorgrids,
ymin=-0.48, ymax=0.48,
ytick style={color=black},
ytick={-0.4, 0, 0.4}
]
\addplot [very thick, orange]
table {%
0 0
0.05 0.00113745
0.1 0.00125119
0.15 0.000227489
0.2 0.00298011
0.25 -0.00232039
0.3 -0.0193366
0.35 -0.0238409
0.4 -0.00648345
0.45 0.0144456
0.5 0.0124437
0.55 0.0042313
0.6 0.00909957
0.65 0.00621046
0.7 -0.020565
0.75 -0.0573046
0.8 -0.0959323
0.85 -0.139064
0.9 -0.189863
0.95 -0.249829
1 -0.301082
1.05 -0.321533
1.1 -0.317757
1.15 -0.301332
1.2 -0.268392
1.25 -0.169616
1.3 -0.108581
1.35 -0.0416078
1.4 0.0293234
1.45 0.10121
1.5 0.169798
1.55 0.233836
1.6 0.291823
1.65 0.3292
1.7 0.344851
1.75 0.363733
1.8 0.382114
1.85 0.377155
1.9 0.349742
1.95 0.30611
2 0.251262
2.05 0.190773
2.1 0.135538
2.15 0.0874242
2.2 0.0495472
2.25 0.0194048
2.3 -0.00450429
2.35 -0.0251603
2.4 -0.0458164
2.45 -0.0985484
2.5 -0.126825
2.55 -0.151098
2.6 -0.167159
2.65 -0.169161
2.7 -0.151553
2.75 -0.115633
2.8 -0.0710449
2.85 -0.0281632
2.9 0.0075299
2.95 0.0334637
3 0.0480685
3.05 0.0550297
3.1 0.0572591
3.15 0.0573501
3.2 0.0573501
3.25 0.0574183
3.3 0.0587832
3.35 0.0640155
3.4 0.0739795
3.45 0.0887663
3.5 0.103917
3.55 0.115314
3.6 0.119841
3.65 0.089972
3.7 0.0500932
3.75 -0.000227489
3.8 -0.0569633
3.85 -0.11718
3.9 -0.169912
3.95 -0.204376
4 -0.227353
4.05 -0.246416
4.1 -0.257268
4.15 -0.256039
4.2 -0.250193
4.25 -0.242504
4.3 -0.232449
4.35 -0.221029
4.4 -0.209677
4.45 -0.198803
4.5 -0.188247
4.55 -0.17776
4.6 -0.166477
4.65 -0.153919
4.7 -0.140179
4.75 -0.126097
4.8 -0.113017
4.85 -0.100937
4.9 -0.106533
4.95 -0.121934
5 -0.144228
5.05 -0.169912
5.1 -0.203489
5.15 -0.243709
5.2 -0.27715
5.25 -0.291619
5.3 -0.294508
5.35 -0.295122
5.4 -0.294735
5.45 -0.293461
5.5 -0.291118
5.55 -0.288525
5.6 -0.285044
5.65 -0.281063
5.7 -0.277219
5.75 -0.27549
5.8 -0.276718
5.85 -0.280199
5.9 -0.284817
5.95 -0.288707
6 -0.289799
6.05 -0.305814
6.1 -0.33582
6.15 -0.356817
6.2 -0.350356
6.25 -0.338322
6.3 -0.333727
6.35 -0.329837
6.4 -0.325014
6.45 -0.32233
6.5 -0.321693
6.55 -0.321693
6.6 -0.321715
6.65 -0.321715
6.7 -0.321715
6.75 -0.321738
6.8 -0.321442
6.85 -0.320896
6.9 -0.320623
6.95 -0.321215
7 -0.323194
7.05 -0.324855
7.1 -0.324696
7.15 -0.327448
7.2 -0.343031
7.25 -0.380544
7.3 -0.369488
7.35 -0.354224
7.4 -0.345761
7.45 -0.339437
7.5 -0.333886
7.55 -0.331702
7.6 -0.333431
7.65 -0.339937
7.7 -0.348514
7.75 -0.35643
7.8 -0.363414
7.85 -0.369147
7.9 -0.373856
7.95 -0.378383
};
\addplot [very thick, green]
table {%
0 0
0.05 0.000409481
0.1 0.00334409
0.15 -0.00366258
0.2 -0.022294
0.25 -0.0246598
0.3 -0.00389007
0.35 0.0197461
0.4 0.0207243
0.45 0.0119659
0.5 0.0133991
0.55 0.00762089
0.6 -0.0207243
0.65 -0.0590107
0.7 -0.104031
0.75 -0.159561
0.8 -0.223941
0.85 -0.281245
0.9 -0.314686
0.95 -0.323376
1 -0.314072
1.05 -0.289412
1.1 -0.248327
1.15 -0.196301
1.2 -0.135879
1.25 -0.00122844
1.3 0.0662676
1.35 0.133877
1.4 0.196551
1.45 0.257131
1.5 0.308066
1.55 0.335752
1.6 0.347126
1.65 0.365666
1.7 0.377109
1.75 0.361412
1.8 0.327084
1.85 0.279084
1.9 0.223781
1.95 0.16973
2 0.123026
2.05 0.0840118
2.1 0.0516401
2.15 0.0251603
2.2 0.00154693
2.25 -0.0237499
2.3 -0.0515718
2.35 -0.0845578
2.4 -0.148528
2.45 -0.170208
2.5 -0.177601
2.55 -0.166886
2.6 -0.13704
2.65 -0.0985256
2.7 -0.0623776
2.75 -0.0351016
2.8 -0.0181991
2.85 -0.00903133
2.9 -0.00454979
2.95 -0.00332134
3 -0.00309385
3.05 -0.00147868
3.1 0.00404931
3.15 0.0147868
3.2 0.0305746
3.25 0.0519586
3.3 0.0738203
3.35 0.0914507
3.4 0.101688
3.45 0.100277
3.5 0.0827379
3.55 0.0486372
3.6 -0.0489557
3.65 -0.106078
3.7 -0.163656
3.75 -0.205855
3.8 -0.21839
3.85 -0.204695
3.9 -0.174598
3.95 -0.137449
4 -0.0989579
4.05 -0.0597387
4.1 -0.0244096
4.15 -0.000295736
4.2 0.0102598
4.25 0.0130351
4.3 0.0128304
4.35 0.010692
4.4 0.00698392
4.45 -0.0031166
4.5 -0.0161972
4.55 -0.0224304
4.6 -0.0238409
4.65 -0.0238409
4.7 -0.0238409
4.75 -0.020929
4.8 -0.0168797
4.85 -0.0185859
4.9 -0.0186769
4.95 -0.0195641
5 -0.0219072
5.05 -0.0242959
5.1 -0.0247736
5.15 -0.0251148
5.2 -0.0263888
5.25 -0.0274807
5.3 -0.0275717
5.35 -0.0285499
5.4 -0.0304836
5.45 -0.0306201
5.5 -0.032895
5.55 -0.0361936
5.6 -0.0372628
5.65 -0.0372628
5.7 -0.0375812
5.75 -0.037695
5.8 -0.0378997
5.85 -0.037968
5.9 -0.038059
5.95 -0.0381727
6 -0.0381727
6.05 -0.0386049
6.1 -0.0389462
6.15 -0.0389234
6.2 -0.0389234
6.25 -0.0395831
6.3 -0.0400609
6.35 -0.0400609
6.4 -0.0400609
6.45 -0.0405841
6.5 -0.0410846
6.55 -0.0410846
6.6 -0.0416988
6.65 -0.0425178
6.7 -0.0439054
6.75 -0.0459073
6.8 -0.0470448
6.85 -0.04866
6.9 -0.0517311
6.95 -0.054097
7 -0.0568041
7.05 -0.0605349
7.1 -0.0639245
7.15 -0.0667454
7.2 -0.0711814
7.25 -0.0793938
7.3 -0.0822829
7.35 -0.0857407
7.4 -0.0881521
7.45 -0.0902223
7.5 -0.0939986
7.55 -0.0975247
7.6 -0.100368
7.65 -0.104258
7.7 -0.106965
7.75 -0.108831
7.8 -0.111311
7.85 -0.11536
7.9 -0.118977
7.95 -0.12157
};
\path [draw=black, draw opacity=0.7, semithick, dash pattern=on 5.55pt off 2.4pt]
(axis cs:0,0.4)
--(axis cs:7.95,0.4);

\path [draw=black, draw opacity=0.7, semithick, dash pattern=on 5.55pt off 2.4pt]
(axis cs:0,-0.4)
--(axis cs:7.95,-0.4);

\end{axis}

\end{tikzpicture}%
        \begin{tikzpicture}

\definecolor{darkgray176}{RGB}{176,176,176}
\definecolor{green}{RGB}{0,128,0}
\definecolor{orange}{RGB}{255,165,0}

\begin{axis}[
height=\figheight/3 - 2/3*0.0cm,
label style={font=\footnotesize},
legend style={font=\tiny, inner sep=0.1pt, outer sep=0pt},
tick align=outside,
tick label style={font=\footnotesize},
tick pos=left,
width=\figwidth,
x grid style={darkgray176},
xlabel={\(\displaystyle t\) [s]},
xmajorgrids,
xmin=0, xmax=8,
xtick style={color=black},
y grid style={darkgray176},
ylabel={\(\displaystyle \theta\) [\si{\degree}]},
ymajorgrids,
ymin=-360, ymax=360,
ytick style={color=black},
ytick={-360, -180, 0, 180, 360}
]
\addplot [very thick, orange]
table {%
0 179.999847960504
0.05 180.088083460954
0.1 180.088083460954
0.15 179.999847960504
0.2 180.263408546264
0.25 179.648624832089
0.3 178.154350902388
0.35 177.890790316628
0.4 179.648624832089
0.45 181.581784432861
0.5 181.494121890205
0.55 180.967000718685
0.6 181.40645934755
0.65 180.878765218235
0.7 178.330248945493
0.75 174.990477957675
0.8 171.650134012063
0.85 168.135037938935
0.9 164.531133407562
0.95 160.751903790879
1 158.642846147003
1.05 160.048884576254
1.1 164.091674778697
1.15 169.980534997052
1.2 177.802554816178
1.25 196.259817228522
1.3 205.664346477799
1.35 214.541181597761
1.4 222.363201416887
1.45 228.603284763657
1.5 232.73431046655
1.55 234.843941068222
1.6 234.579807524667
1.65 230.976475951089
1.7 224.824055146974
1.75 217.353258456263
1.8 207.861066664331
1.85 195.732696057002
1.9 180.967000718685
1.95 164.794693993323
2 148.44706218265
2.05 132.803022544398
2.1 119.443365635332
2.15 108.545135414149
2.2 100.458982051468
2.25 94.5701218331132
2.3 91.1426883026406
2.35 90.1755355444597
2.4 91.5821469315059
2.45 101.425561851853
2.5 109.775275800295
2.55 120.498180936168
2.6 133.593704301678
2.65 149.238316897725
2.7 167.519681266965
2.75 187.734205236975
2.8 208.388760793646
2.85 227.549042420616
2.9 244.160234817049
2.95 258.4864715265
3 269.824160376549
3.05 278.876893539616
3.1 285.556435515251
3.15 289.863359261249
3.2 291.972416905126
3.25 291.797091819816
3.3 289.511563175039
3.35 284.853416300625
3.4 277.82207823878
3.45 268.418121947298
3.5 256.728637011039
3.55 242.665960887348
3.6 226.14243103357
3.65 186.415829350379
3.7 163.828114192937
3.75 141.328061578149
3.8 119.795161721543
3.85 100.371319508813
3.9 83.8477896550349
3.95 69.2579923598236
4 56.865259025821
4.05 46.3183665246119
4.1 36.9140664584526
4.15 28.2129001984778
4.2 20.1269760189146
4.25 12.9199118013025
4.3 5.88868833101606
4.35 -0.87890579857478
4.4 -7.64646555069791
4.45 -14.6777463167639
4.5 -22.1484284159156
4.55 -30.5859194985712
4.6 -40.2538947420497
4.65 -51.4160293236703
4.7 -64.4239474423049
4.75 -79.2773053232764
4.8 -96.4161918490247
4.85 -137.548831961466
4.9 -161.10369987709
4.95 -186.152268764619
5 -211.201410609944
5.05 -235.107501653982
5.1 -257.255758182559
5.15 -276.943160981049
5.2 -294.082047506798
5.25 -309.638424602395
5.3 -323.964661311846
5.35 -337.060757635151
5.4 -349.189701200275
5.45 -360.879186136534
5.5 -372.656333615448
5.55 -384.872939723228
5.6 -398.232596632293
5.65 -412.734158427054
5.7 -429.170025738177
5.75 -447.363154606966
5.8 -468.105372706292
5.85 -490.429527277975
5.9 -514.863312451329
5.95 -539.824218796203
6 -564.082105926452
6.05 -609.346344699582
6.1 -629.474352042528
6.15 -646.611519694891
6.2 -662.608501334943
6.25 -676.760558874675
6.3 -688.712458481103
6.35 -699.082994572971
6.4 -708.576905238289
6.45 -717.452021484866
6.5 -725.977633476412
6.55 -734.766806053719
6.6 -743.991426555325
6.65 -754.190075308654
6.7 -765.614853743563
6.75 -778.449108354493
6.8 -793.300174404284
6.85 -810.265454718108
6.9 -829.688723973042
6.95 -851.042860997568
7 -874.070034783876
7.05 -897.624329741704
7.1 -921.092681030263
7.15 -943.85629423081
7.2 -965.388048171827
7.25 -1001.69065407132
7.3 -1016.63339336833
7.35 -1029.6395353178
7.4 -1040.27363199543
7.45 -1048.97686090346
7.5 -1056.44250097402
7.55 -1062.68774094094
7.6 -1067.78133573966
7.65 -1072.00403468977
7.7 -1075.60793922114
7.75 -1079.03422683603
7.8 -1082.54645812018
7.85 -1086.24203589877
7.9 -1090.37306160166
7.95 -1094.76764789032
};
\addplot [very thick, green]
table {%
0 179.999847960504
0.05 180.088083460954
0.1 180.351644046715
0.15 179.560389331639
0.2 177.978452859283
0.25 177.890790316628
0.3 179.999847960504
0.35 182.197141104831
0.4 182.285376605281
0.45 181.581784432861
0.5 181.670019933311
0.55 180.79110267558
0.6 178.066688359733
0.65 174.46278382836
0.7 170.419993625917
0.75 165.849509294158
0.8 161.10369987709
0.85 157.851591431927
0.9 157.588030846167
0.95 160.136547118909
1 164.882929493773
1.05 171.826032055168
1.1 180.52754208982
1.15 190.283294467512
1.2 200.21494488831
1.25 218.408073757099
1.3 225.439411818944
1.35 230.712915365329
1.4 233.789125767386
1.45 234.931603610877
1.5 233.349667138521
1.55 228.603284763657
1.6 221.660182202262
1.65 213.486366296925
1.7 203.027594704607
1.75 189.843835838647
1.8 174.55101932881
1.85 158.378712603448
1.9 142.382876878985
1.95 128.144302712189
2 116.27949269062
2.05 107.050861484448
2.1 100.107185965257
2.15 95.7131726343992
2.2 93.6912045753825
2.25 94.0430006615928
2.3 96.7674149774399
2.35 101.777357938064
2.4 118.652110920257
2.45 130.693391942726
2.5 145.459087281043
2.55 162.509738306341
2.6 181.933580519071
2.65 201.972779403771
2.7 220.868927487186
2.75 237.304794798309
2.8 251.103337378444
2.85 262.353363685838
2.9 270.96663822004
2.95 277.29495706726
3 281.074186683943
3.05 282.568460613644
3.1 281.689543355913
3.15 278.437434910751
3.2 272.724472735501
3.25 264.638892330615
3.3 254.267783280952
3.35 241.523483043857
3.4 226.49422711978
3.45 209.267678051377
3.5 189.492039752437
3.55 167.95913989583
3.6 125.507550938997
3.65 106.523167355132
3.7 89.8243124160445
3.75 76.113432378564
3.8 63.8962533129894
3.85 52.5586217587202
3.9 42.0996209832841
3.95 32.958996094444
4 25.224638817973
4.05 18.4570332292263
4.1 12.6562366239833
4.15 8.43749108265553
4.2 5.88868833101606
4.25 4.30664299667878
4.3 3.25195374655787
4.35 2.46093967375616
4.4 1.9335950486957
4.45 2.02148677446877
4.5 2.54882566995127
4.55 2.54882566995127
4.6 2.19726449643695
4.65 1.75781159714956
4.7 1.49414214940831
4.75 0.966797524347848
4.8 0.175781159714956
4.85 0.175781159714956
4.9 0.087890579857478
4.95 0
5 0.087890579857478
5.05 0.175781159714956
5.1 0.087890579857478
5.15 0
5.2 0
5.25 0
5.3 -0.087890579857478
5.35 -0.087890579857478
5.4 0
5.45 -0.087890579857478
5.5 0
5.55 0.175781159714956
5.6 0.175781159714956
5.65 0.087890579857478
5.7 0.087890579857478
5.75 0
5.8 0
5.85 -0.087890579857478
5.9 -0.087890579857478
5.95 -0.087890579857478
6 -0.087890579857478
6.05 -0.175781159714956
6.1 -0.175781159714956
6.15 -0.175781159714956
6.2 -0.263671739572434
6.25 -0.263671739572434
6.3 -0.263671739572434
6.35 -0.263671739572434
6.4 -0.351562319429912
6.45 -0.351562319429912
6.5 -0.351562319429912
6.55 -0.43945289928739
6.6 -0.43945289928739
6.65 -0.527343479144868
6.7 -0.527343479144868
6.75 -0.43945289928739
6.8 -0.527343479144868
6.85 -0.527343479144868
6.9 -0.43945289928739
6.95 -0.43945289928739
7 -0.43945289928739
7.05 -0.263671739572434
7.1 -0.263671739572434
7.15 -0.263671739572434
7.2 -0.175781159714956
7.25 0
7.3 0
7.35 0.087890579857478
7.4 0
7.45 -0.087890579857478
7.5 0
7.55 0
7.6 0
7.65 0.087890579857478
7.7 0.087890579857478
7.75 0
7.8 -0.087890579857478
7.85 0.087890579857478
7.9 0.087890579857478
7.95 0.087890579857478
};
\end{axis}

\end{tikzpicture}%
    \end{minipage}%
    \vspace{-0.5cm}
    \caption{Closed-loop evaluation on Quanser and self-made cartpole pendulum hardware:
    \capt{The parameter-adaptive AMPC with true identified system parameters stabilizes both pendulums, Quanser (blue) and self-made (green), while no adaption (red) or parameters from the wrong system instance (orange) fail during the swing-up. 
    Nominal, Quanser, and self-made system have significantly different parameters.
    The slow drift in the cart position for both systems is due to unmodeled stick-slip effects in the pendulum bearing near zero velocity.
    However, as our video shows on a long horizon, the cart slowly oscillates on its rail.
    Notably, with parameter-adaptive AMPC, the systems satisfy constraints.
    }}\label{fig:timeseries}%
\end{figure}

\section{Conclusion}\label{sec:conclusion}
In this work, we presented an approach that allows for tuning parameters of AMPC online without retraining.
Alongside the output of an MPC, parameter-adaptive AMPC also learns the gradient of this output with respect to MPC parameters offline.
Using a linear predictor, we can then adapt the output of the AMPC to parameter changes online.
Our work resulted in the, to the best of our knowledge, first real-world implementation of AMPC on a resource-constrained MCU controlling a fast nonlinear system.
Our simulation and hardware experiments demonstrate the efficacy of parameter-adaptive AMPC.
Throughout experiments, we noticed that tuning is very intuitive due to the physical meaning of dynamics parameters.
Further, identified parameters from real data work well in practice.
As a result, the tuning process for the AMPC took only minutes from first try to a successful swing-up.
Without parameter adaption, this would have required weeks for recomputing datasets.
As such, the method is a significant step towards real-world applications of AMPC.

We identified two significant implementation challenges.
First, despite multiple reinitializations of the NLP solver to find a global optimum, there remain outliers in the dataset due to local minima.
Second, accurate sensitivities require additional fine-tuning of solver settings and high numerical accuracy, thus prolonging dataset generation in practice.
Hence, future work should face these challenges by developing methods that cope with the non-uniqueness of MPC solutions and improve sensitivity calculation.
As the results suggest, parameter-adaptive AMPC enables practical applications for systems where classic AMPC would be too cumbersome to tune, \eg in robotics.
We plan to investigate such new applications and the generalization of our method to high dimensional systems in future research.

\acks{
This work is funded in part by the  German Research Foundation (DFG) – RTG 2236/2 (UnRAVeL) and the DFG priority program 1914 (grant TR 1433/1-2).
Simulations were performed with computing resources granted by RWTH Aachen University under project rwth1500.
We thank Andr{\'e}s Posada Moreno and Christian Fiedler for helpful discussions.
}

\bibliography{references}

\end{document}